\documentclass[12pt]{article}

\usepackage{amsmath}
\usepackage{amsthm}
\usepackage{amsfonts}
\usepackage{amssymb}
\usepackage{latexsym} 
\usepackage{epsfig}
\usepackage{graphicx}

\usepackage{calc}  
\usepackage[matrix,tips,graph,curve,web]{xy}

\makeatletter

\def\@maketitle{%
  \newpage
  \null
  \let \footnote \thanks
    {\normalfont\sffamily\bfseries\Large\noindent\@title \par}%
    \vskip 1em%
    {\normalfont\sffamily 
        \noindent
        \@author
        \par}
  \par
  \vskip 4em}
\def\@seccntformat#1{\csname the#1\endcsname{.\ }}
\renewcommand\section{\@startsection {section}{1}{\z@}%
                                   {-3.0ex \@plus -1ex \@minus -.2ex}%
                                   {1.5ex \@plus.2ex}%
                                   {\normalfont\large\bfseries}}
\renewcommand\subsection{\@startsection{subsection}{2}{\z@}%
                                     {-2.75ex\@plus -1ex \@minus -.2ex}%
                                     {1.5ex \@plus .2ex}%
                                   {\normalfont\large}}
\def\fnum@figure{\normalfont\footnotesize\figurename~\thefigure}
\renewcommand\tableofcontents{%
    \section*{\contentsname
        \@mkboth{%
           \MakeUppercase\contentsname}{\MakeUppercase\contentsname}}%
    \@starttoc{toc}%
    }
\renewcommand*\l@part[2]{%
  \ifnum \c@tocdepth >-2\relax
    \addpenalty\@secpenalty
    \addvspace{2.25em \@plus\p@}%
    \begingroup
      \setlength\@tempdima{3em}%
      \parindent \z@ \rightskip \@pnumwidth
      \parfillskip -\@pnumwidth
      {\leavevmode
       \large \bfseries #1\hfil \hb@xt@\@pnumwidth{\hss #2}}\par
       \nobreak
       \if@compatibility
         \global\@nobreaktrue
         \everypar{\global\@nobreakfalse\everypar{}}%
      \fi
    \endgroup
  \fi}
\renewcommand*\l@section[2]{%
  \ifnum \c@tocdepth >\z@
    \addpenalty\@secpenalty
    \addvspace{1.0em \@plus\p@}%
    \setlength\@tempdima{1.5em}%
    \begingroup
      \parindent \z@ \rightskip \@pnumwidth
      \parfillskip -\@pnumwidth
      \leavevmode \sffamily\bfseries
      \advance\leftskip\@tempdima
      \hskip -\leftskip
      #1\nobreak\hfil \nobreak\hb@xt@\@pnumwidth{\hss #2}\par
    \endgroup
  \fi}
\renewcommand*\l@subsection{\sffamily\@dottedtocline{2}{1.5em}{2.3em}}
\renewcommand*\l@subsubsection{\@dottedtocline{3}{3.8em}{3.2em}}
\renewcommand*\l@paragraph{\@dottedtocline{4}{7.0em}{4.1em}}
\renewcommand*\l@subparagraph{\@dottedtocline{5}{10em}{5em}}
\makeatother

\makeatletter 
\@addtoreset{equation}{section}
\makeatother

\renewcommand{\theequation}{\thesection.\arabic{equation}}

\theoremstyle{plain}
\newtheorem{theorem}[equation]{Theorem}
\newtheorem{corollary}[equation]{Corollary}
\newtheorem{lemma}[equation]{Lemma}
\newtheorem{proposition}[equation]{Proposition}

\newtheorem*{theorem*}{}

\theoremstyle{definition}
\newtheorem{definition}[equation]{Definition}

  {\begin{list}{}%
    {%
    \settowidth{\labelwidth}{#1}%
    \setlength{\itemindent}{0pt}%
    \setlength{\labelsep}{1em}%
    \setlength{\leftmargin}{\labelwidth+\parindent+\labelsep}%
    \setlength{\itemsep}{0pt}%
    \setlength{\parsep}{.6ex}}}%
  {\end{list}}

  {\begin{list}{}%
    {%
    \settowidth{\labelwidth}{#1}%
    \setlength{\itemindent}{0pt}%
    \setlength{\labelsep}{1em}%
    \setlength{\leftmargin}{\labelwidth+\labelsep}%
    \setlength{\itemsep}{0pt}%
    \setlength{\parsep}{.6ex}}}%
  {\end{list}}

\makeatletter
\newenvironment{subequations*}{
  \begingroup 
  \let\protect\@nx
  \edef\@tempa{\def\@nx\theparentequation{\theequation}}%
  \@xp\endgroup\@tempa
  \setcounter{parentequation}{\value{equation}}%
  \setcounter{equation}{0}%
  \def\theequation{\theparentequation\alph{equation}}%
  \ignorespaces
}{%
  \setcounter{equation}{\value{parentequation}}%
  \global\@ignoretrue
}

\renewcommand\det{{\rm det\,}}

\def\d/{/\mspace{-6.0mu}/}

\usepackage{helvet}
\usepackage{graphicx,amssymb,amsmath,amsfonts,amsthm,color,url,placeins,float}
\usepackage{wrapfig}
\usepackage{tabstackengine}
\usepackage{multirow,ctable} 

\renewcommand\section{\@startsection{section}{1}{\z@}%
                                   {-3.0ex \@plus -1ex \@minus -.2ex}%
                                   {1.5ex \@plus.2ex}%
                                   {\normalfont\sffamily\large\bfseries}}
\renewcommand\subsection{\@startsection{subsection}{2}{\z@}%
                                     {-2.75ex\@plus -1ex \@minus -.2ex}%
                                     {1.5ex \@plus .2ex}%
                                   {\normalfont\sffamily\large}}
\renewcommand\subsubsection{\@startsection{subsubsection}{3}{\z@}%
                                     {-2.75ex\@plus -1ex \@minus -.2ex}%
                                     {1.5ex \@plus .2ex}%
                                   {\normalfont\sffamily\large}}

\newcommand{\od}{\stackrel{\mbox {\tiny {def}}}{=}}

\def\RR{\mathbb{R}}

\def\d{\mathrm{d}}

\def\RR{\mathbb{R}}

\def\RR{\mathbb{R}}

\def\det{\operatorname{det}}

\def\max{\mathrm{max}}

\def\supp{\operatorname{supp}}

\def\od{\stackrel{\mathrm{def}}{=}}

\def\supp{\operatorname{supp}}

\def\FP{\operatorname{FP}}

\def\Wtil{\widetilde{W}}

\usepackage{color}
\definecolor{cherry}{rgb}{0.9,.1,.2}

\begin{document}

\noindent {\Large \bf Robust motifs of threshold-linear networks}\\
Carina Curto$^1$, Christopher Langdon$^1$, Katherine Morrison$^2$
\bigskip

\noindent {\small $^1$ Department of Mathematics, The Pennsylvania State University, University Park, PA 16802\\
$^2$ School of Mathematical Sciences, University of Northern Colorado, Greeley, CO 80639}
\medskip 

\vspace{.25in}

\noindent {\bf Abstract.} To any inhibition-dominated threshold-linear network (TLN) we can associate a directed graph that captures the pattern of strong and weak inhibition between neurons. Robust motifs are graphs for which the structure of fixed points in the network is independent of the choice of connectivity matrix $W$, and whose dynamics are thus greatly constrained. This makes them ideal building blocks for constructing larger networks whose behaviors are robust to changes in connection strengths. In contrast, flexible motifs correspond to networks with multiple dynamic regimes. In this work, we give a purely graphical characterization of both flexible and robust motifs of any size. We find that all but a few robust motifs fall into two infinite families of graphs with simple feedforward architectures.

\tableofcontents

\section{Introduction}\label{sec:intro}

Network dynamics are determined by a variety of factors including intrinsic dynamics of single neurons, neuromodulators, and the structure of network connectivity. Here we focus on the role of connectivity, which can be broken down into two components: network architecture and synaptic weights. The architecture is typically represented by a directed graph $G$, while synaptic weights can be thought of as continuous parameters given as entries of a real-valued weight matrix $W$. The architecture, however, imposes strong constraints on the weights. By considering all $W$ compatible with a graph $G$, we obtain the range of dynamic possibilities that can arise from $G$. It is thus natural to ask: how does the architecture of a network constrain the allowed dynamics? Are some architectures fundamentally more flexible -- or constraining -- than others?

In this work we address these questions in the context of threshold-linear networks (TLNs), whose dynamics are given by
\begin{equation}\label{eq:dynamics}
\dfrac{dx_i}{dt} = -x_i + \left[\sum_{j=1}^n W_{ij}x_j+\theta \right]_+, \quad i = 1,\ldots,n.
\end{equation}
Here $n$ is the number of neurons, $x_i(t) \in \RR_{\geq 0}$ is the firing rate of neuron $i$, and $\theta >0$ is a constant external input (the same for each neuron).  The values $W_{ij}$ are entries of an $n \times n$ matrix of real-valued connection strengths, and the threshold-nonlinearity is given by $[\cdot]_+ = \max\{0,\cdot\}$. 
We will also use the standard notation $[n] \od \{1,\ldots,n\}$ to denote the set of all neurons (or nodes) of the network.

TLNs exhibit the full repertoire of nonlinear dynamic behavior, including multistability, limit cycles, quasiperiodic attractors, and chaos \cite{CTLN-preprint}. Much of this behavior is shaped by the set of stable and unstable fixed points \cite{book-chapter, CTLN-paper}. Because the system~\eqref{eq:dynamics} is piecewise linear, for generic $W$ there can be at most one fixed point for each linear region of the state space. In particular, there is at most one fixed point per possible support, where the {\it support} of a fixed point is the subset of neurons with positive firing rate \cite{fp-paper}.  A fixed point where all neurons fire is called a {\it full-support} fixed point, and has support $[n]$. We denote the set of all fixed point supports as
$$\FP(W) \od \{\sigma \subseteq [n] \mid \sigma \text{ is the support of a fixed point}\}.$$
Once the allowed supports $\FP(W)$ are known, the corresponding fixed points are easily recovered \cite{fp-paper}.
Note that the supports depend only on $W$: the choice of $\theta$ simply rescales the values of the fixed points, while keeping the supports intact. Prior work has focused on the set of {\it stable} fixed points of~\eqref{eq:dynamics} in the case of {\it symmetric} $W$ \cite{Seung-Nature, XieHahnSeung, HahnSeungSlotine, flex-memory, net-encoding, pattern-completion}, while more recent work has highlighted the importance of {\it unstable} fixed points in shaping the dynamics of these networks \cite{book-chapter, CTLN-paper}. In particular, new mathematical tools have been developed to investigate the relationship between the full set $\FP(W)$ and the structure of the connectivity matrix $W$ \cite{fp-paper}.

In this article, as in \cite{fp-paper}, we restrict ourselves to {\it competitive} (i.e., inhibition-dominated) networks, so that $W$ is effectively inhibitory with $W_{ij} < 0$ for each $i\neq j,$ and $W_{ii}=0$.  To each such network, we associate a directed connectivity graph $G_W$ as follows:
\vspace{-.03in}
\begin{equation}\label{eq:TLN-graph}
j \to i \;\text{ in }\; G_W \, \,  \Leftrightarrow \, \,  W_{ij} >-1. \quad \quad 
\end{equation}
To motivate the definition of $G_W$, note that there is an implicit ``leak'' time constant in the $-x_i$ term of~\eqref{eq:dynamics}, which has been set to $1$. An edge $j \to i$, corresponding to $W_{ij}>-1$, thus indicates that $j$ inhibits $i$ less than $j$ inhibits itself through its self-leak term. On the other hand, the absence of an edge, corresponding to $W_{ij}<-1$, signifies that the inhibition from $j$ to $i$ is stronger than the leak term.   For this reason, even though all interactions are effectively inhibitory, the activity of a network often appears to follow the edges of the graph. 

For a given architecture $G$, only a subset of competitive TLNs will have $W$ satisfying $G_W = G$. What do they have in common? In other words, how does the architecture constrain the network dynamics? To address this question, we will use the set of fixed point supports $\FP(W)$ as a proxy for qualitative features of the dynamics. This enables us to define different types of graph motifs. First, we consider whether or not the full-support fixed point is always present or always absent for a given $G$. This allows us to partition directed graphs into three categories: invariant permitted, invariant forbidden, and flexible motifs. In the following definition, and elsewhere in this article, the phrase ``all $W$ such that $G_W = G\,$'' refers to $W$ within the family of TLNs that are both competitive and {\it nondegenerate} (see Section~\ref{sec:background} for precise details).

\begin{definition}
Let $G$ be a directed graph. We say that $G$ is an {\it invariant permitted motif} if $[n] \in \FP(W)$ for all $W$ such that $G_W = G$. If $[n] \not\in \FP(W)$ for all $W$ such that $G_W = G$, then $G$ is an {\it invariant forbidden motif}. Finally, if there exist $W_1, W_2$ with graph $G$ such that $[n] \in \FP(W_1)$ and $[n] \notin \FP(W_2)$, then $G$ is a {\it flexible motif}. 
\end{definition}

\begin{figure}[!h]
\begin{center}
\includegraphics[width=\textwidth]{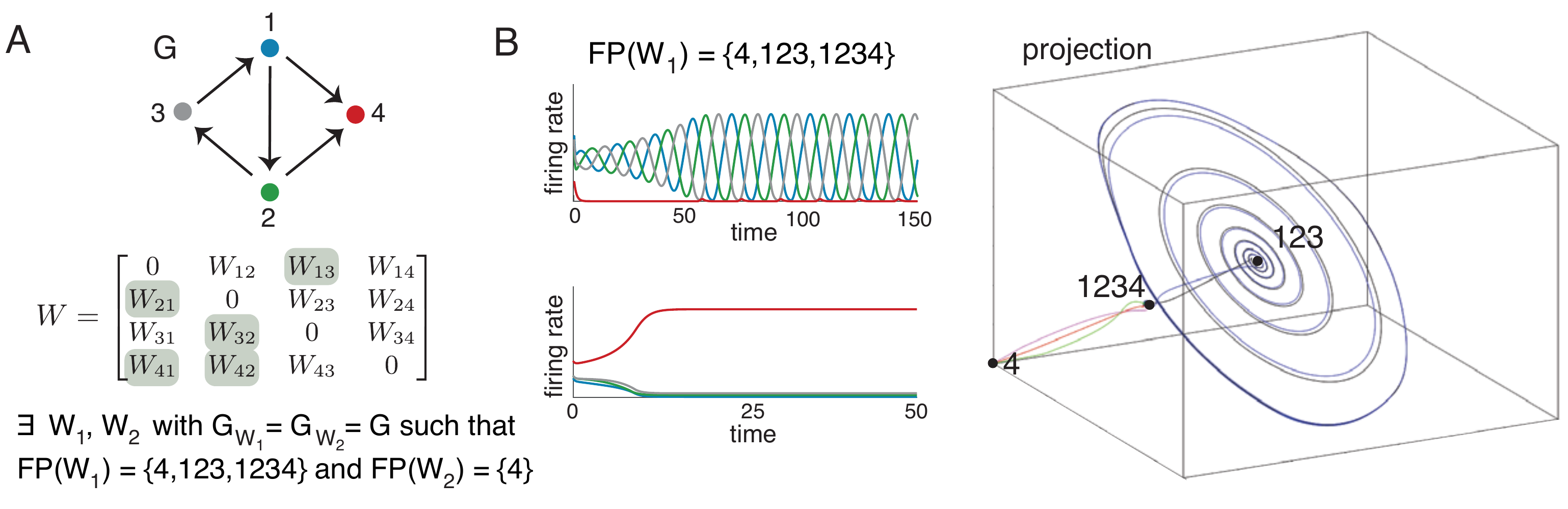}
\caption{{A flexible motif.} (A) The graph $G$ determines which off-diagonal $W$ entries are above or below $-1$. The entries above $-1$, corresponding to edges in the graph, are shaded in the matrix below. Two example matrices, $W_1$ and $W_2$, have the same graph $G$ but different sets of fixed point supports. In $W_1$, all unshaded entries were set to $-1.25,$ and the shaded entries are $W_{21}=W_{32}=W_{13}=-.81, W_{41}=W_{42}=-.999$. In $W_2$, the values of the edges coming into node $4$ have been increased to $W_{41}=W_{42}=-.9$, and all other entries are the same as those of $W_1$. This change is enough to eliminate the fixed points supported on $123$ and $1234$, leaving only the stable fixed point supported on $4$. (B) In the case of $W_1$, the network has two attractors. The first is a limit cycle corresponding to the unstable fixed point for $123$ (top), and the second is the stable fixed point for $4$ (bottom). The fixed point for $1234$ corresponds to a ``tipping point'': initial conditions near this point will converge to one attractor or the other. (C) A projection of several four-dimensional trajectories onto the $x_4 = 0$ subspace. Five initial conditions were chosen near the tipping point $1234$. Three of these trajectories (red, green, violet) evolved towards the stable fixed point $4$, while the other two (black, blue) initially approached the unstable fixed point $123$ before falling into the corresponding limit cycle attractor.}
\label{fig:flexible-motif}
\end{center}
\vspace{-.2in}
\end{figure}

The graph in Figure~\ref{fig:flexible-motif} is an example of a flexible motif. As shown in panel A, the associated matrices are constrained to have $W_{ij} > -1$ for the shaded entries, while the remaining off-diagonal entries all have  
$W_{ij} < -1$. For some choices of $W$, there are three fixed points with supports $\FP(W) = \{4, 123, 1234\}$, while other choices yield a single fixed point supported on the sink node $4$. In the first dynamic regime, the first two fixed points correspond to attractors (see Figure~\ref{fig:flexible-motif}B), while the full-support fixed point $1234$ serves as a kind of ``tipping point." Indeed, Figure~\ref{fig:flexible-motif}C depicts trajectories of the corresponding TLN in a projection of the $4$-dimensional state space. The three fixed points are labeled by their supports, and different colors correspond to different trajectories. Five distinct initial conditions were selected near the $1234$ fixed point: three converged to the stable fixed point $4$ while two initially approached the unstable fixed point $123$ before spiraling out and falling into the limit cycle (the same one shown at the top of panel B). This nicely illustrates the tipping point behavior of the full-support fixed point. In the second dynamic regime, where $\FP(W) = \{4\},$ all trajectories converge to a single stable fixed point. 

Note that every graph $G$ is either invariant permitted, invariant forbidden, or flexible. 
Surprisingly, there is a simple graph-based characterization for each of these families. This is given in Theorem~\ref{thm:flexible-motifs}, below. Here we need the notions of source and target of a graph, as previously defined in \cite{fp-paper}. A {\it source} is a node with no incoming edges, while a {\it target} is a node with incoming edges from all other nodes. The {\it singleton} is the graph with a single node and no edges. A {\it clique} is a complete graph with all bidirectional edges, and an {\it independent set} is a graph with no edges.

\begin{theorem}\label{thm:flexible-motifs}
Let $G$ be a graph on $n$ nodes. Then
\begin{enumerate}
\item $G$ is a flexible motif if and only if $n \geq 3$, $G$ does not contain both a source and a target, and $G$ is not the $3$-cycle.
\item $G$ is an invariant forbidden motif if and only if $G$ contains both a source and a target, and $G$ is not the singleton.
\item $G$ is an invariant permitted motif if and only if $G$ is either the singleton, the $2$-clique, the independent set of size $2$, or the $3$-cycle.
\end{enumerate}
\end{theorem}

Flexible motifs are guaranteed to have multiple dynamic regimes. At the other extreme, there are graphs that completely determine $\FP(W)$. We refer to these as {\it robust motifs}, because the fixed point structure is independent of the choice of $W$ and therefore robust to arbitrary perturbations of the weights that fix $G_W$, provided the TLN stays competitive and nondegenerate.

\begin{definition}
Let $G$ be a directed graph on $n$ nodes. We say that $G$ is a \emph{robust motif} if $\FP(W)$ is identical for all $W$ such that $G_W = G$. 
\end{definition}

\noindent Robust motifs cannot be flexible: they are either invariant permitted or invariant forbidden.
For ease of comparison, all four motif definitions are summarized in the table below.

\begin{table}[H]
\begin{center}
\begin{tabular}{l|l}
type of motif $G$ & definition\\
\hline
\hline
flexible & $\exists\; W_1, W_2$ with $G_{W_1} = G_{W_2} = G$ s.t. $[n] \in \FP(W_1)$ and $[n] \notin \FP(W_2)$\\
\hline
invariant permitted & $[n] \in \FP(W)$ for all $W$ with $G_W = G$\\
\hline
invariant forbidden & $[n] \notin \FP(W)$ for all $W$ with $G_W = G$\\
\hline
robust & $\FP(W)$ is identical for all $W$ with $G_W = G$\\
\hline
\hline
\end{tabular}
\end{center}
\end{table}
\vspace{-.15in}

At first glance, it is not obvious that robust motifs should even exist. Moreover, while it is relatively easy to conclude that a motif is {\it not} robust (e.g., by exhibiting two different $W$ matrices that give rise to distinct $\FP(W)$), it is not so straightforward to prove that a given motif {\it is} robust. A brute force computation sampling all $W$ for a given graph $G$ is not possible. 

Despite the apparent difficulties, building on previous results about fixed points of TLNs we are able to provide a complete characterization of robust motifs for all $n$. We find that the bulk of robust motifs fall into two infinite families, with only a handful of (very small) exceptions. The two families are obtained via small modifications of directed acyclic graphs (DAGs), which are graphs that have a feedforward structure.\footnote{Specifically, $G$ is a DAG if and only if there exists an ordering of the vertices such that if $j > i$, then $j \not\to i$.}  We call the families DAG1 and DAG2, and they are defined as follows:
\begin{itemize}
\item DAG1: The graphs in this family consist of all DAGs of size $n>1$ that contain a target node. These can always be decomposed as a DAG of size $n-1$ plus a target node $t$, with no edges from the target back to the DAG. 
\item DAG2: This family consists of all graphs that can be decomposed as a (nonempty) DAG plus a $2$-clique, where at least one of the nodes in the $2$-clique is a target of the full graph and there are no edges from the $2$-clique back to the DAG. 
\end{itemize}
These families are illustrated in Figure~\ref{fig:venn-diagram}B (right). 
Note that all graphs in DAG1 and DAG2 have a source node and a target node, and are thus invariant forbidden motifs by Theorem~\ref{thm:flexible-motifs}. They can thus be thought of as generalizations of the single directed edge graph (see Figure~\ref{fig:venn-diagram}B, bottom left), which is the smallest element of DAG1. Any graph that is a longer path, however, is a DAG that does {\it not} belong to DAG1 because it has no target. On the other hand, all oriented directed cliques (i.e. complete DAGs, having a directed edge between every pair of nodes) belong to the DAG1 family. Directed cliques have previously been found to be overrepresented and dynamically meaningful in large, biologically realistic models of cortical networks \cite{ReimannHessMarkram17}, as well as in experimental and computational studies of mammalian cortex \cite{ChambersMacLean16, DecheryMacLean18, BojanekZhuMacLean19}.

It turns out that DAG1 and DAG2 comprise all robust motifs that are invariant forbidden; the rest are invariant permitted. The following theorem is our main result.

\begin{theorem}\label{thm:robust-motifs}
$G$ is a robust motif if and only if one of the following holds:
\begin{itemize}
\item[(i)] $G$ belongs to one of the infinite families DAG1 or DAG2, or
\item[(ii)] $G$ is one of the four invariant permitted motifs (the singleton, the $2$-clique, the independent set of size $2$, or the $3$-cycle).
\end{itemize}
\end{theorem}

\begin{figure}[!h]
\begin{center}
\includegraphics[width=\textwidth]{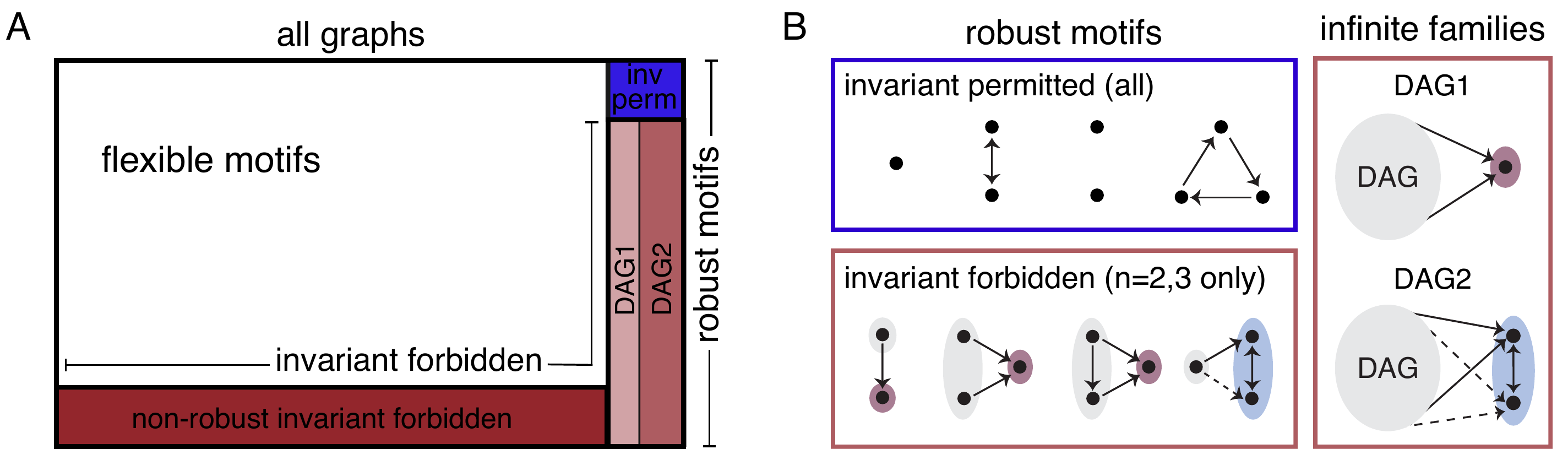}
\vspace{.025in}
\caption{Robust motifs in the zoo of all graph motifs. (A) Schematic of all graphs, and the breakdown into flexible motifs, robust motifs, and invariant permitted and forbidden motifs. (B) Summary of robust motifs. Only four are invariant permitted (top left), while the remainder fall into two infinite families: DAG1 and DAG2 (right). The DAG1 arrows indicate that all nodes in the DAG have edges onto the target node at the right. The dashed arrows in DAG2 represent optional edges, as only the top right node is required to be a target. The first five invariant forbidden motifs are also shown, color-coded according to the family (bottom left). Note that the last graph picture represents two graphs, one where the optional dashed edge is included and one where it is not.}
\label{fig:venn-diagram}
\end{center}
\vspace{-.1in}
\end{figure}

Figure~\ref{fig:venn-diagram}A illustrates the classification of all graphs into the various motif types. The vast majority of graphs are flexible motifs. Among those that are not, only four are invariant permitted and these are all robust (see Figure~\ref{fig:venn-diagram}B). Among the invariant forbidden motifs, those that belong to DAG1 or DAG2 are robust and the rest are not. The smallest robust motifs in the DAG1 and DAG2 families, for $n \leq 3$, are shown at the bottom of Figure~\ref{fig:venn-diagram}B.

Table~\ref{table:number-robust} displays the numbers of robust motifs for small $n$, as well as separate counts for the DAG1 and DAG2 families. For each $n$, the number of graphs in DAG1 is simply given by the number of DAGs of size $n-1$. To count the graphs in DAG2, observe that from any graph in DAG1 one can create one (or more) distinct graphs in DAG2 by adding an edge from the target node back to a sink of the size $n-1$ DAG obtained by removing the target. For $n \leq 6$, we were able to count non-isomorphic DAG2 graphs in this way. For $n \geq 7$, we can only conclude that there are more DAG2 graphs than DAG1 graphs of matching $n$. Figure~\ref{fig:robust-motifs} illustrates all 13 of the robust motifs for $n=4$.

\vspace{.2in}
\begin{table}[H] 
\begin{center}
\begin{small} 
\begin{tabular}{l|ccccc}
$n$\phantom{llll}& \begin{tabular}{c} $\#$ directed\\ graphs \end{tabular} &  $\#$ DAGs &  \begin{tabular}{c} $\#$ graphs\\ in DAG1 \end{tabular} & \begin{tabular}{c} $\#$ graphs\\ in DAG2 \end{tabular} &   \begin{tabular}{c} $\#$ robust\\ motifs \end{tabular}\\
\hline
\hline
1& 1 & 1 & 0 & 0 & 1* \\
\hline
2& 3 & 2 & 1 & 0 & 3$^*$ \\
\hline
3& 16 & 6 & 2 & 2 & 5$^*$ \\
\hline
4& 218  & 31 & 6 & 7 & 13\\
\hline
5& 9,608  & 302 & 31 & 40 & 71\\
\hline
6& 1,540,944  & 5,984 & 302 & 420 & 722\\
\hline
7& 882,033,440  & 243,668 & 5,984 & $>$ 5,984 & $>$11,968\\
\hline
\hline
\end{tabular}
\end{small}
\vspace{.2in}
\caption{Numbers of robust motifs for small $n$. For each graph size, the total numbers of directed graphs and DAGs are taken from the On-Line Encyclopedia of Integer Sequences \cite{num-directed-graphs, num-DAGs}. $^*$Note that the 1,3, and 5 counts include the invariant permitted robust motifs: the singleton, the $2$-clique, the independent set of size $2$, and the $3$-cycle.}
\vspace{-.2in}
\label{table:number-robust}
\end{center}
\end{table}

\begin{figure}[!h]
\vspace{-.1in}
\begin{center}
\includegraphics[width=\textwidth]{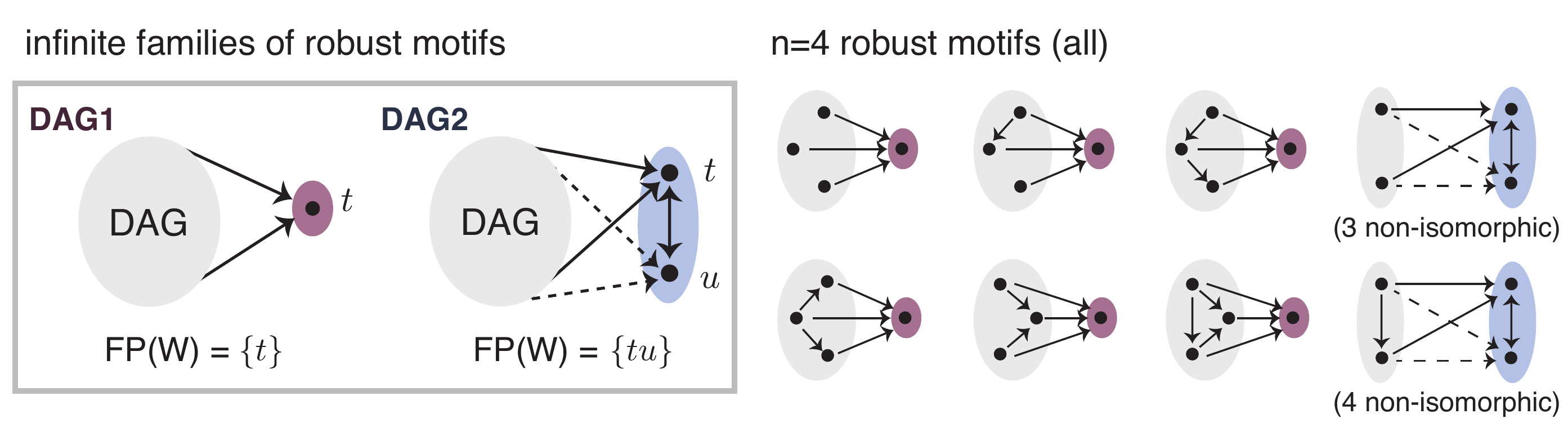}
\caption{Robust motifs for $n=4$. All robust motifs for $n \geq 4$ come from the DAG1 and DAG2 families (left). For $n=4$, there are $13$ robust motifs. Six are from the DAG1 family, and seven are from DAG2 (right).}
\label{fig:robust-motifs}
\end{center}
\vspace{-.2in}
\end{figure}

A key result that we'll use to prove Theorem~\ref{thm:robust-motifs} is a fixed point collapse theorem, below, which applies to graphs with a source and a target. To state the theorem, we must first define source-target decompositions.

\begin{definition}
Let $G$ be a graph on $n$ nodes that contains both a source and a target. A {\it source-target decomposition} of $G$ is any partition of the nodes $\omega \:\dot\cup\: \tau = [n]$ such that $\omega$ contains a source, $\tau$ contains a target, $G|_\omega$ is a DAG, and there are no edges from $\tau$ to $\omega$.
\end{definition}

Every graph that contains a source $s$ and target $t$ has at least one source-target decomposition: one can simply take $\omega=\{s\}$ and $\tau = [n] \setminus \{s\}$.  Recall from Theorem~\ref{thm:flexible-motifs} that the invariant forbidden motifs are precisely the graphs that contain both a source and a target. In particular, the following ``collapse'' result applies to all invariant forbidden motifs. Here the notation $W_\tau$ refers to the principal $|\tau| \times |\tau|$ submatrix of $W$ obtained by restricting to entries indexed by elements of $\tau$.

\begin{theorem}\label{thm:collapse}
Let $G$ be a graph that contains both a source and a target. Then for any source-target decomposition of $G$, with partition $\omega \:\dot\cup\: \tau = [n]$, we have
$$\FP(W) = \FP(W_\tau)$$
for every $W$ with graph $G_W = G$. In particular, every fixed point of a competitive TLN with graph $G$ has support contained in $\tau$.
\end{theorem}

Although the theorem holds for any source-target decomposition, it is clearly most informative when $\tau$ is chosen to be as small as possible. It turns out that this choice is unique, and corresponds to the choice where $G|_\tau$ has no sources.  To see this, observe that if $\omega \:\dot\cup\: \tau = [n]$ is source-target decomposition and $G|_\tau$ contains a source node $v$ (meaning $v$ receives no edges from $\tau$), then one can obtain a new decomposition by moving $v$ from $\tau$ to $\omega$. It is easy to see that the new partition still satisfies all the properties of a source-target decomposition. This process can be continued until $G|_\tau$ contains no remaining sources (see Figure~\ref{fig:source-target-decomposition-example}). Finally, it is straightforward to show that the minimal $\tau$ obtained in this way is unique.

\begin{figure}[!h]
\vspace{-.05in}
\begin{center}
\includegraphics[width=.9\textwidth]{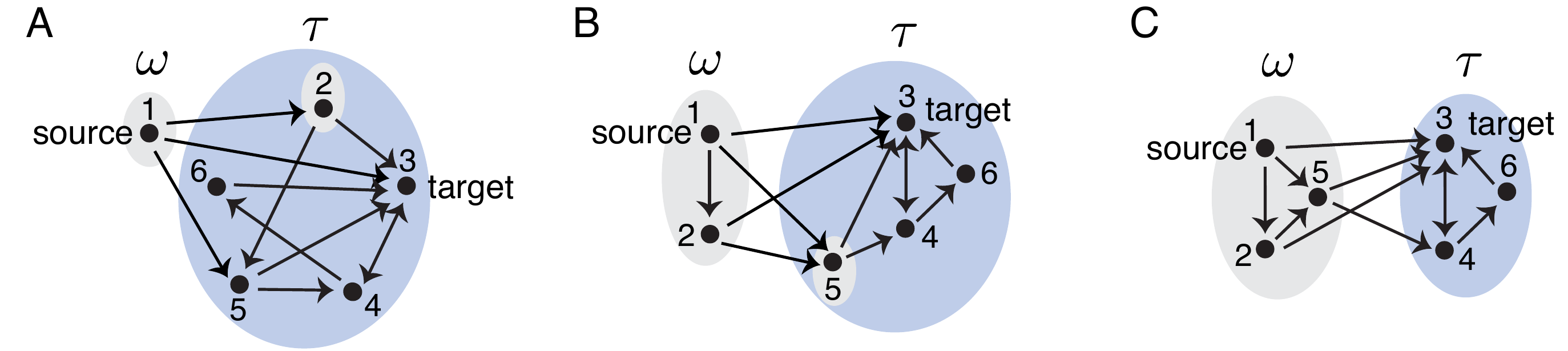}
\vspace{.1in}
\caption{Iterative construction of source-target decompositions. (A) A source-target decomposition of a graph where $\omega$ contains only a single source.  The gray highlighted node $2$ is a source in $G|_\tau$, but not a source in the full graph.  (B) A second source-target decomposition is obtained by moving node $2$ to $\omega$. Now node $5$ has become a source in the new $G|_\tau$. (C) A third source-target decomposition is obtained by moving $5$ to $\omega$. In this decomposition, $G|_\tau$ has no sources. We have thus arrived at the unique source-target decomposition with minimal $\tau$.}
\label{fig:source-target-decomposition-example}
\end{center}
\vspace{-.2in}
\end{figure}

Theorem~\ref{thm:collapse} allows us to compute $\FP(W)$ for the robust motifs in DAG1 and DAG2, as these are all invariant forbidden. This is because DAG1 and DAG2 have obvious source target decompositions with $\tau = \{t\}$ and $\tau = \{t,u\}$, respectively, using the notation for the nodes in Figure~\ref{fig:robust-motifs}. Theorem~\ref{thm:collapse} then tells us that $\FP(W) = \FP(W_{\{t\}}) = \{t\}$ for any $W$ corresponding to a graph in DAG1. For graphs in DAG2, the theorem gives $\FP(W) = \FP(W_{\{t,u\}})$, where $\{t,u\}$ is the $2$-clique. In \cite{fp-paper} this was shown to have a unique fixed point of full support, and so we have $\FP(W) = \{tu\}$ for every $W$ with graph in DAG2. In both the DAG1 and DAG2 cases, the fixed points collapse to those of a singleton and a clique, respectively. It then follows from results in \cite{fp-paper} that the unique fixed point is also {\it stable}. This is true for all graphs belonging to the infinite families DAG1 and DAG2, and also for the singleton and the $2$-clique in isolation (also robust motifs, see Figure~\ref{fig:venn-diagram}B). We summarize these observations in the following corollary.

\begin{corollary}
Let $G$ be a robust motif that is not the $3$-cycle nor the independent set of size $2$. Then for all $W$ with graph $G_W = G$, the corresponding competitive TLN has a unique fixed point, and this fixed point is stable.
\end{corollary}

Based on this fixed point structure, we expect that all but two of the robust motifs will yield extremely boring network dynamics. Namely, the activity should always converge to the stable fixed point. In the case of the $3$-cycle, the typical behavior we expect is that of a limit cycle resembling the one depicted in Figure~\ref{fig:flexible-motif}B (top). Finally, the independent set of size $2$ has two stable fixed points, supported on each of the singletons, and a full support fixed point which is unstable and acts as a tipping point. Here we expect standard WTA dynamics, with the network converging to one stable fixed point or the other depending on initial conditions.

\begin{figure}[!h]
\begin{center}
\includegraphics[width=.95\textwidth]{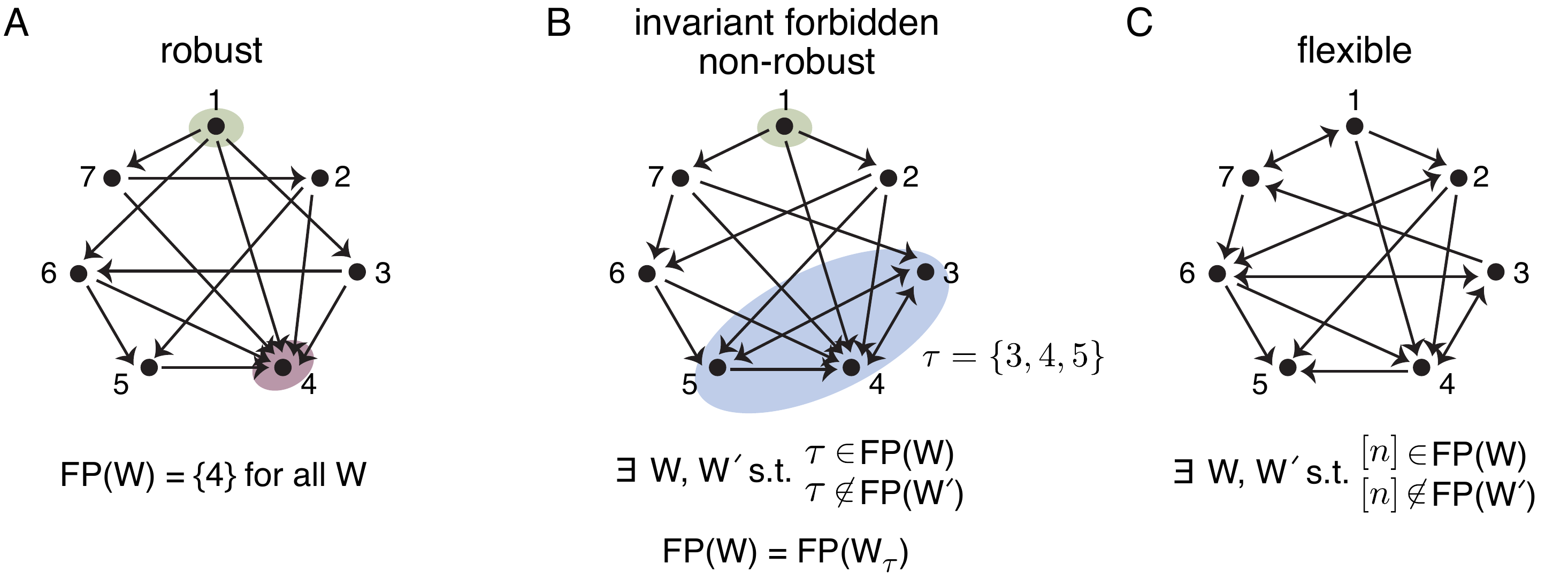}
\caption{Graphical analysis of TLN networks. (A) A graph with a source (node 1) and a target (node 4) that has a source-target decomposition where $\tau$  is precisely the target node. Clearly, this graph belongs to the DAG1 family and is thus a robust motif, with a single fixed point supported on the target. (B) A graph with a source (node 1) and a target (node 4) such that the source-target decomposition with minimal $\tau$ has $\tau =\{3,4,5\}$.  It follows that $\FP(W)=\FP(W_\tau)$ for all $W$ with graph $G$. Moreover, $G|_\tau$ is flexible so there exist $W, W'$ such that  $\tau \in \FP(W_\tau)$ and $\tau \notin \FP(W_\tau')$.  (C) A graph that does not have both a source and a target, and is thus flexible.
}
\label{fig:large-examples}
\end{center}
\vspace{-.1in}
\end{figure}

Putting together Theorems~\ref{thm:flexible-motifs},~\ref{thm:robust-motifs}, and~\ref{thm:collapse} allows us to determine from a purely graphical analysis whether any given motif is flexible, invariant permitted, invariant forbidden, or robust. Furthermore, in the case of invariant forbidden motifs (having a source and a target), we can reduce the computation of $\FP(W)$ to that of a subnetwork, $\FP(W_\tau)$. Figure~\ref{fig:large-examples} illustrates the kinds of conclusions we can draw for a set of three example networks that differ in only a small number of edges.

The organization of the rest of this article is as follows. In Section~\ref{sec:background}, we review some basic background about fixed points of TLNs. In Section~\ref{sec:inv-motifs}, we develop the theory related to invariant permitted, invariant forbidden, and flexible motifs, and use it to prove Theorem~\ref{thm:flexible-motifs}. Finally, in Section~\ref{sec:robust} we prove Theorems~\ref{thm:robust-motifs} and~\ref{thm:collapse}.

\FloatBarrier

\section{Preliminaries}\label{sec:background}

A {\it fixed point} $x^*$ of a TLN $(W,\theta)$ is a point in the state space of~\eqref{eq:dynamics} satisfying 
$\left.\dfrac{dx_i}{dt}\right|_{x = x^*} = 0$ for each $i \in [n]$, where $[n] \od \{1,\ldots,n\}$.
In other words,
\begin{equation}\label{eq:x*}
\quad x_i^* = \left[\sum_{j=1}^n W_{ij}x_j^*+\theta \right]_+,  \quad \text{for all } i \in [n].
\end{equation}
The {\it support} of a fixed point $x^*$ is the subset of active neurons,
$$\supp(x^*) \od \{i \in [n] \mid x_i^*>0\}. \quad$$
We typically refer to supports as subsets $\sigma \subseteq [n]$.  We denote the set of all fixed point supports of a TLN $(W,\theta)$ as
$$\FP(W) \od \{\sigma \subseteq [n] \mid \sigma \text{ is the support of a fixed point}\}.$$ 
Note that we omit $\theta$ from the notation $\FP(W)$ because the value of $\theta >0$ has no impact on the fixed point supports, it simply scales the precise value of the corresponding fixed points. 

We will often restrict matrices and vectors to a particular subset of neurons $\sigma$.
We use the notation $A_\sigma$ and $b_\sigma$ to denote a matrix $A$ and a vector $b$ that have been truncated to include only entries with indices in $\sigma$. Furthermore, we use the notation $(A_i;b)$ to denote a matrix $A$ whose $i$th column has been replaced by the vector $b$, as in Cramer's rule. In the case of a restricted matrix, $((A_\sigma)_i;b_\sigma)$ denotes the matrix $A_\sigma$ where the column corresponding to the index $i \in \sigma$ has been replaced by $b_\sigma$ (note that this is not typically the $i$th column of $A_\sigma$).

\begin{definition} \label{def:nondegenerate}
We say that a TLN $(W,\theta)$ is {\it nondegenerate} if 
\begin{itemize}
\item $\det(I-W_\sigma) \neq 0$ for each $\sigma \subseteq [n]$, and 
\item for each $\sigma \subseteq [n]$ and all $i \in \sigma$,  the corresponding Cramer's determinant is nonzero: $\det((I-W_\sigma)_i;\theta) \neq 0$. 
\end{itemize}
\end{definition}
\noindent Note that almost all TLNs are nondegenerate, since having a zero determinant is a highly fine-tuned condition.  
Throughout the following, we will restrict consideration to TLNs that are both {\bf competitive}, so that $W_{ij} < 0$ for $i \neq j$ and $W_{ii}=0$ for all $i \in [n]$, and {\bf nondegenerate}.

Observe that when a TLN $(W, \theta)$ is nondegenerate, each fixed point is completely determined by its support, and thus it suffices to study the collection of supports, $\FP(W)$, rather than the specific fixed point values.  In particular, for $\sigma \in \FP(W)$, let 
\begin{equation}\label{eq:xsigma}
x^\sigma \od \theta (I-W_\sigma)^{-1} 1_\sigma, \quad
\end{equation}
where $W_\sigma$ is the principal submatrix of $W$ obtained by truncating $W$ to just the rows and columns indexed by $\sigma$, and $1_\sigma$ refers to the all-ones vector similarly truncated to the indices in $\sigma$.  Then the vector $x^*$ given by $x^*_i=x^\sigma_i$ for all $i \in \sigma$ and $x^*_k=0$ for all $k \notin \sigma$ is the fixed point of $(W, \theta)$ with support $\sigma$.  

This formula for $x^\sigma$ also gives a straightforward way to check if a given $\sigma$ is the support of a fixed point of a TLN $W$, i.e.\ if $\sigma \in \FP(W)$.  It follows from Equation~\ref{eq:x*} that $\sigma \in \FP(W)$ if and only if the following two \emph{fixed point conditions} hold:
\begin{itemize}
\item[(1)] $x_i^\sigma > 0$ for all $i \in \sigma$, and
\item[(2)]  $\sum_{i\in\sigma} W_{ki}x_i^\sigma+\theta \leq 0$ for all $k \notin \sigma$. 
\end{itemize}
(This is straightforward, but see \cite{pattern-completion} for more details.) 

Observe that the first fixed point condition only depends on the values of the restricted submatrix $W_\sigma$, and thus a necessary condition for $\sigma \in \FP(W)$ is that $\sigma \in \FP(W_\sigma)$.  We say that $\sigma$ is a \emph{permitted motif of $W$} if $\sigma \in \FP(W_\sigma)$; otherwise, we say $\sigma$ is a \emph{forbidden motif of $W$}.   In general, the permittedness of a motif will depend on the specific choice of $W$.  However, in the special case when $\sigma \in \FP(W|_\sigma)$ for every $W_\sigma$ with graph $G|_\sigma$, then we say that $G|_\sigma$ is an \emph{invariant permitted motif}.  On the other hand if $\sigma \not\in \FP(W|_\sigma)$ for every $W_\sigma$ with graph $G|_\sigma$, then we say that $G|_\sigma$ is an \emph{invariant forbidden motif}.

When $\sigma$ is a permitted motif of $W$, it produces a fixed point of the full network $W$ precisely it \emph{survives} the addition of all external nodes $k \notin \sigma$.  Fixed point condition (2) shows that survival can be checked one external node at a time, as each inequality to be checked for $k \notin \sigma$ depends only on the entries of $W_{\sigma \cup k}$.  These observations are summarized in the following lemma, which first appeared in \cite{fp-paper}.

\begin{lemma}[Corollary 2 in \cite{fp-paper}]\label{lemma:inheritance} 
Let $(W,\theta)$ be a TLN on $n$ neurons, and let $\sigma \subseteq [n]$.
The following are equivalent:

\begin{enumerate}
\item $\sigma \in \FP(W)$
\item $\sigma \in \FP(W_\tau)$ for all $\sigma \subseteq \tau \subseteq [n]$.
\item $\sigma \in \FP(W_\sigma)$ and $\sigma \in \FP(W_{\sigma \cup \{k\}})$ for all $k \notin \sigma$
\item $\sigma \in \FP(W_{\sigma \cup \{k\}})$ for all $k \notin \sigma$
\end{enumerate}
\end{lemma}

Another key property of $\FP(W)$ is \emph{parity}:

\vspace{-.05in}
\begin{lemma}[parity \cite{CTLN-paper, fp-paper}]\label{lemma:parity}
 Let $(W, \theta)$ be a TLN.  Then the total number of fixed points, $|\FP(W)|,$ is odd.  
 \end{lemma}
 
\vspace{-.05in} 
\noindent Parity is particularly useful for determining when a network has a full-support fixed point.  If we know which of the smaller subsets are permitted and forbidden motifs, and we know which of the permitted motifs survive, then parity immediately tells us whether or not $W$ has a full-support fixed point; this will be the key to proving that the $3$-cycle always has a full-support fixed point and is in fact robust. In general, though, determining the survival of permitted motifs, in order to apply parity, is nontrivial.


\section{Invariant permitted, invariant forbidden, and flexible motifs}\label{sec:inv-motifs}
In this section we develop the tools needed to prove Theorem~\ref{thm:flexible-motifs} providing a complete characterization of invariant permitted, invariant forbidden, and flexible motifs.  We begin by reviewing some prior results on how connectivity constrains $\FP(W)$ for particularly small graphs.  In \cite[Appendix A.4]{fp-paper}, $\FP(W)$ was determined for all graphs of size $n \leq 2$; these findings are summarized in Figure~\ref{fig:n2-graphs}.  From this, we see that all the graphs up through size 2 are robust.  Moreover, since the singleton, the independent set of size $2$, and the $2$-clique all have full-support fixed points (see panels A-C), these are invariant permitted motifs, while the directed edge (panel D) is invariant forbidden.  In Section~\ref{sec:inv-proofs}, we prove that the $3$-cycle (panel E) always has $\FP(W) = \{123\}$, and thus is invariant permitted and robust.  

\begin{figure}[!h]
\begin{center}
\includegraphics[width=6.5in]{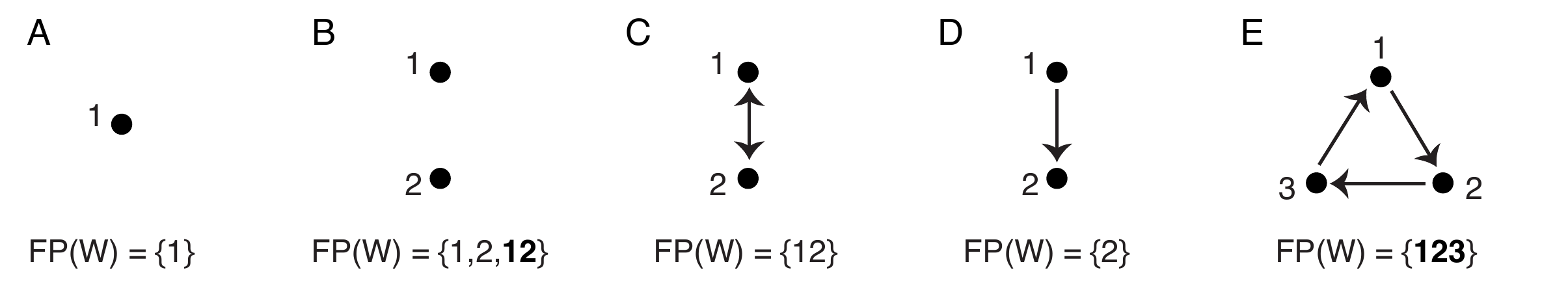}
\caption{(A-D) For each directed graph on $n \leq 2$ nodes, the allowed fixed point supports $\FP(W)$ are completely constrained by the structure of the associated graph $G_W$. (E) The $3$-cycle graph also has a unique set $\FP(W)$, consisting of a single fixed point support. Note, additionally, that each fixed point here is always stable, with the exception of the $12$ fixed point in panel B (supported on the independent set) and the fixed point of panel E (supported on the $3$-cycle), both shown in bold. For larger graphs, most fixed points are unstable and stability can depend on the choice of $W$.}
\label{fig:n2-graphs}
\end{center}
\vspace{-.2in}
\end{figure}

\begin{figure}[!h]
\begin{center}
\includegraphics[width=6.75in]{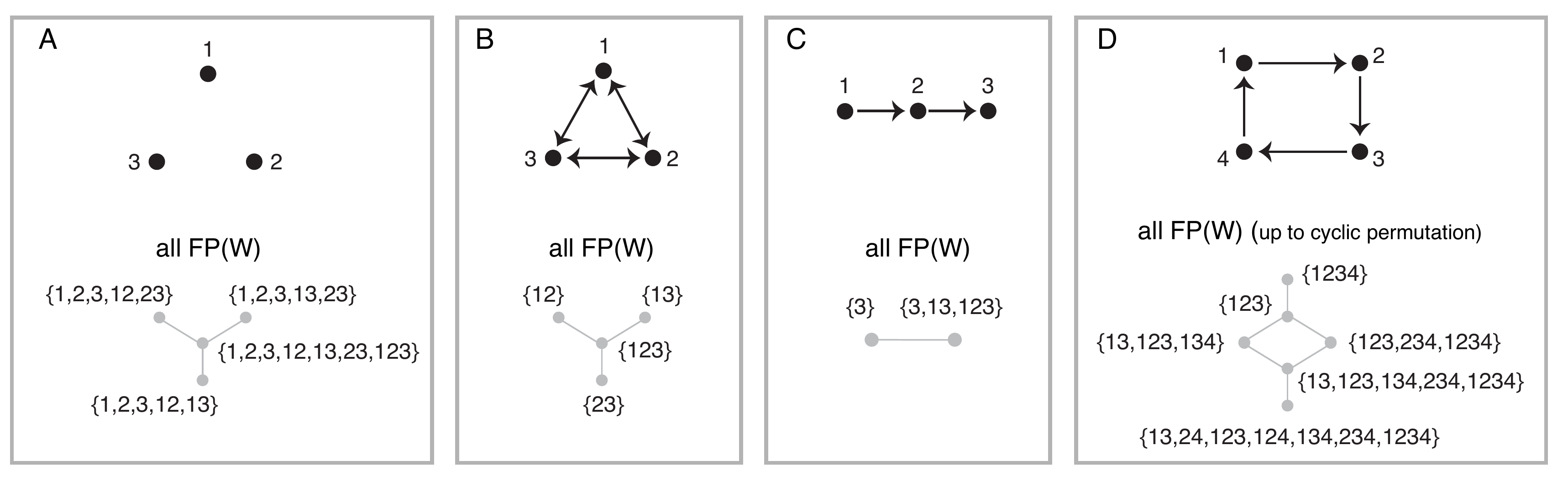}
\caption{Graphs with multiple possibilities for $\FP(W)$. (A) The independent set of size $3$ always contains all three singletons $1, 2$ and $3$ in $\FP(W)$, but the larger fixed point supports depend on the choice of $W$. (B) The $3$-clique always supports a unique fixed point, but the support can either be one of the $2$-cliques or the full $3$-clique. (C) The path of length $2$ has two possibilities for $\FP(W)$. (D) The $4$-cycle has 16 distinct dynamic regimes for $\FP(W)$. Here we show only 6: the remainder can be obtained via cyclic permutation. Note that for all four graphs, A-D, the presence of the full-support fixed point is parameter-dependent, and thus these graphs are all flexible motifs.}
\label{fig:n3-graphs}
\end{center}
\vspace{-.2in}
\end{figure}

Interestingly, none of the natural generalizations of these graphs are robust, and in fact they are all flexible.  For example, both the independent set and clique on $n=3$ nodes have four distinct $\FP(W)$ regimes, as shown in Figure~\ref{fig:n3-graphs}A-B, the path of length 2 has two possibilities for $\FP(W)$ (see Figure~\ref{fig:n3-graphs}C), and each of these graphs has at least one $W$ regime with a full-support fixed point and one without. Notice however that within the class of flexible motifs there is still a significant range in terms of how much the graph constrains the set of possible $\FP(W)$.  For example, the path of length 2 has only two possibilities for $\FP(W)$, while the $4$-cycle has sixteen different possibilities for $\FP(W)$, ranging from a single fixed point of full support to a collection of $7$ fixed points with a variety of supports (see Figure~\ref{fig:n3-graphs}D).  The extent of the graphical constraints on $\FP(W)$ can be somewhat explained by the number of proper subgraphs that are invariant permitted or invariant forbidden; this is explored further in \cite{oriented-matroids-paper}, where the machinery of oriented matroids is used to work out all the possible $\FP(W)$ regimes for small competitive TLNs.

In the remainder of Section~\ref{sec:inv-motifs}, we develop the machinery needed to prove Theorem~\ref{thm:flexible-motifs}.  We begin in Section~\ref{sec:targets-sources-subgraphs} by considering subgraphs with targets and sources and determining when these can support surviving fixed points in the full network in Proposition~\ref{prop:target-source}.  In Section~\ref{sec:strong-domination}, we introduce the notion of \emph{strong domination}, which will be key to the proof of Proposition~\ref{prop:target-source}.  Finally, in Section~\ref{sec:inv-proofs}, we combine these results with Propositions~\ref{prop:fix-pt-construction} and~\ref{prop:node-complement}, which give explicit constructions of $W$ with and without full-support fixed points, in order to prove Theorem~\ref{thm:flexible-motifs}.  

\subsection{Targets and sources of subgraphs}\label{sec:targets-sources-subgraphs}

We begin by defining internal and external targets and sources of subgraphs.

\begin{definition}[targets and sources] A node $k$ of $G$ is a {\it target of} $G|_\sigma$ if $i \to k$ for all $i \in \sigma \setminus k$. We call $k$ an {\it internal target} if $k \in \sigma$, and an {\it external target} if $k \notin \sigma$. Similarly, we say that $j$ is a {\it source of} $G|_\sigma$ if $i \not\to j$ for all $i \in \sigma \setminus j$. We call $j$ an {\it internal source} if $j \in \sigma$, and an {\it external source} if $j \notin \sigma$. 
\end{definition}

\begin{proposition}\label{prop:target-source}
Let $G$ be a graph on $n$ nodes, and $\sigma \subseteq [n]$ nonempty. Then for any $W$ with graph $G_W = G,$ the following statements hold:
\begin{itemize}
\item[(i)] If $G|_\sigma$ has an internal source and an internal target, then $\sigma \notin \FP(W_\sigma)$ and thus $\sigma \notin \FP(W)$.
\item[(ii)] If $G|_\sigma$ has an internal source and an external target, then $\sigma \notin \FP(W)$.
\item[(iii)] If $G|_\sigma$ has an internal target and no outgoing edges (i.e., all nodes outside $\sigma$ are external sources) then $\sigma \in \FP(W) \; \Leftrightarrow \; \sigma \in \FP(W_\sigma)$.
\item[(iv)] If $G|_\sigma$ has an external source, $j \notin \sigma,$ and an external target, $k \notin \sigma$, then
$\sigma \in \FP(W_{\sigma \cup k}) \Rightarrow \sigma \in \FP(W_{\sigma \cup j}).$
\end{itemize}
\end{proposition}

We prove Proposition~\ref{prop:target-source} in the following section by appealing to a new concept known as \emph{strong domination}. We will then use it in the proofs of Theorem~\ref{thm:flexible-motifs}.

Proposition~\ref{prop:target-source} has several important corollaries. Taking $\sigma = [n]$, part (i) immediately yields the following result.

\begin{corollary}\label{cor:source-target}
If $G$ contains both a source and a target, then $G$ is an invariant forbidden motif.
\end{corollary}

Since cliques always contain an internal target, we obtain another corollary when we combine part (iii) of Proposition~\ref{prop:target-source} with the previous finding that the $2$-clique is invariant permitted \cite[Appendix A.4]{fp-paper}.

\begin{corollary}\label{cor:cliques}
Let $G$ be a graph on $n$ nodes and $\sigma \subsetneq [n]$.  If $G|_\sigma$ is a clique with no outgoing edges in $G$, then for any $W$ with graph $G_W=G$, 
$$\sigma \in \FP(W) \ \Leftrightarrow \ \sigma \in \FP(W_\sigma).$$
In particular, if $\sigma$ is the $2$-clique, then $\sigma \in \FP(W)$ for every $W$ with graph $G_W=G$.
\end{corollary}

Note that a singleton $i$ is trivially both a source and an internal target in its restricted subgraph $G|_{\{i\}}$.  Thus, combining parts (ii) and (iii) of Proposition~\ref{prop:target-source}, we immediately obtain the following characterization of precisely when a singleton is a fixed point support.  (This result was previously shown in \cite[Appendix A.4]{fp-paper}, but by a very different method of proof).  For the next corollary, recall that a \emph{sink} is a node in a graph that has no outgoing edges.  

\begin{corollary}\label{cor:singleton}
Let $G$ be a graph and $i$ a node in $G$.  For any $W$ with graph $G_W=G$, 
$$ \{i \} \in \FP(W) \ \ \Leftrightarrow \ \ i \text{ is a sink in } G.$$
\end{corollary}

From Corollaries~\ref{cor:cliques} and~\ref{cor:singleton}, we see that a sink and a 2-clique with no outgoing edges will always support fixed points in the full network irrespective of the choice of $W$.  In Appendix A.2, we show that these are the {\it only} motifs that yield surviving fixed points independent of parameters.


\subsection{Strong domination and proof of Proposition~\ref{prop:target-source}}\label{sec:strong-domination}

Here we introduce strong domination, a useful tool for ruling in and ruling out fixed point supports in TLNs. 

\begin{definition}[strong domination]\label{def:strong-domination}
Consider a TLN on $n$ nodes with matrix $W,$ let $\sigma \subseteq [n]$ be nonempty, and let $\Wtil \od -I + W$.
For any $j,k \in [n]$, we say that $k$ \emph{strongly dominates} $j$ with respect to $\sigma$, and write $k \gg_\sigma j$, if 
$$\Wtil_{ki} \geq \Wtil_{ji} \quad \text{ for all } i \in \sigma,$$ 
and the inequality is strict for at least one such $i$.  
\end{definition}

Strong domination has some immediate implications for the associated graph of the TLN. To see this, 
recall that $\Wtil_{i\ell} = W_{i\ell}$ whenever $i \neq \ell$, and $\Wtil_{ii} = -1$. Now observe that if $\Wtil_{ki} \geq \Wtil_{ji}$, then $i \to j$ implies $i \to k$ since $ \Wtil_{ji} > -1$ implies $\Wtil_{ki} > -1$. Furthermore, since $\Wtil_{jj} = -1$, the inequality $\Wtil_{kj} \geq \Wtil_{jj}$ implies $j \to k$, while $\Wtil_{kk} \geq \Wtil_{jk}$ implies $k \not\to j$. It follows that if $k \gg_\sigma j$, so that $\Wtil_{ki} \geq \Wtil_{ji}$ for each $i \in \sigma$, then the following conditions on the graph $G = G_W$ must hold:
\begin{enumerate}
\item for all  $i \in \sigma \setminus \{j, k\}$, if $i \to j$ then $i \to k$,
\item if $j \in \sigma$, then $j \to k,$ and
\item if $k \in \sigma$, then $k \not\to j$.
\end{enumerate}
In prior work \cite{fp-paper}, we defined {\it graphical domination} using the above three conditions. Specifically, we say that $k$ {\it graphically dominates} $j$ {\it with respect to} $\sigma$ precisely when conditions 1-3 all hold. Graphical domination is thus a necessary (but not sufficient) condition for strong domination. In the special case where $k$ graphically dominates $j$ because $k$ is a target of $G|_\sigma$ and $j$ is a source, then the corresponding strong domination $k \gg_\sigma j$ is guaranteed to hold.

\begin{lemma}\label{lemma:source-target-domination}
Let $G$ be a graph on $n$ nodes, $\sigma \subseteq [n]$ nonempty, and $j,k \in [n]$ with $j \neq k$. If $j$ is a source of $G|_\sigma$ (internal or external), and $k$ is a target of $G|_\sigma$ (internal or external), then $k \gg_\sigma j$.
\end{lemma}

\begin{proof}
If $j$ is a source of $G|_\sigma$, then $W_{ji} < -1$ for all $i \in \sigma\setminus j$, and thus $\Wtil_{ji} \leq -1$ for all $i \in \sigma$. Similarly, if $k$ is a target of $G|_\sigma$, then $W_{ki} > -1$ for all $i \in \sigma\setminus k$, and thus $\Wtil_{ki} \geq -1$ for all $i \in \sigma$. If follows that $\Wtil_{ki} \geq \Wtil_{ji}$ for all $i \in \sigma$. Note, moreover, that the inequality is strict for any $i \in \sigma \setminus \{j,k\}$. If no such $i$ exists, then at least one of $j,k \in \sigma$, since $\sigma$ is nonempty. The strict inequality then follows from choosing $i = j$ (if $j \in \sigma$) or $i = k$ (if $k \in \sigma$). We conclude that in all cases $k \gg_\sigma j$.
\end{proof}

In prior work \cite{fp-paper}, we completely characterized fixed points of competitive TLNs in terms of a related, but more difficult to check, notion called \emph{general domination}.  In Appendix A.1, we show that the presence of strong domination implies general domination, and thus we obtain the following lemma as an immediate consequence of Lemma~\ref{lemma:domination} in Appendix A.1. This will be our key tool for ruling in and ruling out various fixed point supports.

\begin{lemma} \label{lemma:strong-domination}
Suppose $k \gg_\sigma j$. Then the following statements all hold:
\begin{itemize}
\item[(i)] If $j,k \in \sigma$, then $\sigma \notin \FP(W_\sigma)$, and thus $\sigma \notin \FP(W)$.
\item[(ii)] If $j \in \sigma$ and $k \notin \sigma$, then $\sigma \notin \FP(W_{\sigma \cup k})$, and thus $\sigma \notin \FP(W)$.
\item[(iii)] If $k \in \sigma$ and $j \notin \sigma$, then $\sigma \in \FP(W_{\sigma \cup j}) \; \Leftrightarrow \; \sigma \in \FP(W_\sigma)$.
\item[(iv)] If $j,k \notin \sigma$, then $\sigma \in \FP(W_{\sigma \cup k}) \Rightarrow \sigma \in \FP(W_{\sigma \cup j}).$
\end{itemize}
\end{lemma}

\noindent Putting together Lemmas~\ref{lemma:source-target-domination} and~\ref{lemma:strong-domination}, we can now easily prove Proposition~\ref{prop:target-source}.

\begin{proof}[{\bf Proof of Proposition~\ref{prop:target-source}}]
In each part (i)-(iv), let $j$ be the source and $k$ the target of $G|_\sigma$. By Lemma~\ref{lemma:source-target-domination}, it follows that $k \gg_\sigma j$ in each part, and the four hypotheses precisely match those in Lemma~\ref{lemma:strong-domination}. Thus the conclusions in each part (i)-(iv) follow immediately from those of Lemma~\ref{lemma:strong-domination}.
\end{proof}

\subsection{Proof of Theorem~\ref{thm:flexible-motifs}}\label{sec:inv-proofs}

In this section, we collect the remaining results necessary to complete the proof of Theorem~\ref{thm:flexible-motifs}.  We begin by proving the characterization of invariant forbidden motifs (part 2 of this theorem).  Proposition~\ref{prop:target-source} established that if $G$ contains both a source and a target, then it is an invariant forbidden motif. This establishes the backwards direction of part 2 of Theorem~\ref{thm:flexible-motifs}. The forward direction follows immediately from the following proposition, which gives an explicit construction for a $W$ with a full-support fixed point whenever $G$ does not contain both a source and a target.  

\begin{proposition}\label{prop:fix-pt-construction}
Let $G$ have size $n>1$. If $G$ does \underline{not} contain both a source and a target, then there exists a $W$ with graph $G_W = G$ for which $\FP(W)$ has a full-support fixed point. In particular, $G$ is not an invariant forbidden motif.
\end{proposition}

\begin{proof}
Let $n$ be the number of nodes in $G$, and $r$ a real number that satisfies $0<|r|<n$. We will construct $W$ with graph $G$ such that 
$$x = \dfrac{\theta}{n+r}\left(\begin{array}{c}1 \\ \vdots \\1\end{array}\right) = \dfrac{\theta}{n+r} {\bf 1},$$
is a fixed point of the associated TLN. (Note that the entries of $x$ are all strictly positive, and thus it has full support.) For $x$ to be a fixed point, we must have $(I-W)x = \theta {\bf 1}$, so $I-W$ must have uniform row sums equal to $R\od n+r.$
Equivalently, since all the off-diagonal entries of $W$ have the form $-1 \pm a_{ij}$, we can write $I-W$ as $11^T+A$, where the matrix $A$ has $A_{ii}=0$ along the diagonal; then, we must construct $A$ with all row sums equal to $r$. If this $A$ is compatible with $G$, then $x$ will be a full-support fixed point for the corresponding $W$, and we are done.

Note that for $A$ to be compatible with $G$, the off-diagonal entries need to satisfy $A_{ij} < 0$ if $j \to i$, and $A_{ij}>0$ if $j \not\to i$. We show how to construct $A$ in two separate cases which cover all graphs that do not contain both a source and a target. (If $G$ contains neither a source nor a target, then both constructions work.)

\underline{Case 1}: G has no targets. Pick $r>0$. For each row $i$, there exists $k$ such that $k \not\to i$ since $i$ is not a target, and thus $A_{ik}>0$. For any $j \neq i,k$, let
$$A_{ij} = \left\{\begin{array}{cc} -r/n & \text{if} \; j \to i,\\ \phantom{-}r/n & \text{if} \; j \not\to i. \end{array}\right.$$
Now observe that $-r < \sum_{j \neq k} A_{ij} < r,$ so we can choose $A_{ik} = r-\sum_{j \neq k} A_{ij}>0.$ It follows that $\sum_{j=1}^n A_{ij} = r$ for each row $i$, as desired.

\underline{Case 2}: G has no sources. Pick $r < 0$. For each row $i$ there exists $k$ such that $k \to i$ since $i$ is not a source, and thus $A_{ik}<0$.
For any $j \neq i,k$, let
$$A_{ij} = \left\{\begin{array}{cc} \phantom{-}r/n & \text{if} \; j \to i,\\ -r/n & \text{if} \; j \not\to i. \end{array}\right.$$
Observe that $r < \sum_{j \neq k} A_{ij} < -r,$ so we can choose $A_{ik} = r-\sum_{j \neq k} A_{ij}<0.$ Again, it follows that $\sum_{j=1}^n A_{ij} = r$ for each row $i$, as desired.
\end{proof}

We can now immediately prove the following proposition characterizing invariant forbidden motifs (this is part 2 of Theorem~\ref{thm:flexible-motifs}).

\begin{proposition}\label{prop:inv-forbidden}
A graph $G$ is an invariant forbidden motif if and only if $G$ contains both a source and a target, and $G$ is not the singleton.
\end{proposition}
\begin{proof}
Recall that the singleton graph ($n=1$) is trivially invariant permitted, since the $1 \times 1$ matrix $I-W=[1]$ always has $x=\theta$ as a fixed point solution. So if $G$ is invariant forbidden, it must have size $n>1$. The forward direction now follows immediately from Proposition~\ref{prop:fix-pt-construction}, above. The backwards direction was given in Corollary~\ref{cor:source-target}, which followed directly from part (i) of Proposition~\ref{prop:target-source}.
\end{proof}

We can also use the ideas in the proof of Proposition~\ref{prop:fix-pt-construction} to show when a graph $G$ cannot be an invariant permitted motif. Specifically, since there can be only one solution to the (nondegenerate) linear system $(I-W)x = \theta {\bf 1}$, if we can find a $W$ compatible with $G$ such that $x$ does {\it not} lie within the positive orthant, then it follows that this $W$ does not have a full-support fixed point, and thus $G$ cannot be invariant permitted. This strategy allows us to rule out invariant permitted motifs in all but a handful of special cases.

To state the next proposition, we need to introduce the notion of the {\it node-complement} of a graph. Given a directed graph $G$ and a choice of node $i$, define $\widehat{G}(i)$ as the graph obtained from $G$ by requiring that $i \to j$ in $\widehat{G}(i)$ if and only if $i \not\to j$ in $G$. In other words, we complement all outgoing edges of $i$ and leave the rest of the graph intact. This has the effect of changing the sign of all the entries in column $i$ of the associated $A$ matrix (defined in the proof of Proposition~\ref{prop:fix-pt-construction} above).

\begin{proposition}\label{prop:node-complement}
Suppose there exists a node $i$ of $G$ such that $\widehat{G}(i)$ does not contain both a source and a target. Then there exists a $W$ with graph $G_W = G$ for which $\FP(W)$ does \underline{not} have a full-support fixed point. In particular, $G$ is not an invariant permitted motif.
\end{proposition}

\begin{proof}
First observe that if $\widehat{G}(i)$ does not contain both a source and a target, then by Proposition~\ref{prop:fix-pt-construction} there exists a $\widehat{W}$ compatible with $\widehat{G}(i)$, such that the associated TLN has a full-support fixed point of the form $\hat{x} = \dfrac{\theta}{n+r}{\bf 1}$. In particular, using $\widehat{A}$ defined via 
$I-\widehat{W} = 11^T + \widehat{A}$, we have $(11^T + \widehat{A})\hat{x} = \theta {\bf 1}.$ Now consider $W$ and $A$ obtained from  $\widehat{W}$ and $\widehat{A}$ by flipping the signs of all entries in column $i$ of $\widehat{A},$ and a vector $y$ obtained by flipping the sign of the $i$th entry of the all-ones vector, ${\bf 1}$, to $-1$. Note that these sign flips align, so that $Ay = \widehat{A}{\bf 1} = r{\bf 1}$, and thus
$$(11^T+A)y = (n-2){\bf 1} + Ay = (n-2+r){\bf 1}.$$
It is thus clear that the vector $x = \dfrac{\theta}{n-2+r}y$ satisfies $(I-W)x = \theta{\bf 1}$, where $W$ is compatible with $G$. Since $x$ does not lie in the positive orthant (the $i$th entry is negative), it follows that the TLN for $W$ does not have a full-support fixed point.
\end{proof}

Proposition~\ref{prop:node-complement} implies that in order for a graph $G$ to be invariant permitted, every node-complement graph $\widehat{G}(i),$ for $i \in [n]$, must contain both a source and a target. As we see in the next lemma, this requirement is incredibly restrictive.

\begin{lemma}\label{lemma:n3}
Suppose $G$ is a graph on $n$ nodes such that for each $i \in [n]$, the node-complement graph $\widehat{G}(i)$ contains both a source and a target. Then $n \leq 2$ or $G$ is the $3$-cycle. 
\end{lemma}

\begin{proof}
Suppose $\widehat{G}(i)$ contains both a source and a target for each $i \in [n]$. First, observe that any node $k$ can be a target in at most one of the $\widehat{G}(i)$ graphs. To see this, observe that if $k$ is a target in $\widehat{G}(i)$ then either $i = k$ or $i \to k$ in $\widehat{G}(i)$. In the first case, $k$ is a target in all of $G$ and thus $j \not\to k$ in any other $\widehat{G}(j)$. In the second case, $i \not\to k$ in $G$ and thus $i \not\to k$ in any $\widehat{G}(j)$ for $j \neq i$. Since there are $n$ node-complement graphs, it follows that each $k \in [n]$ must be a target in exactly one of them. In particular, this means the in-degree of each node of $G$ is at least $n-2$. By a similar argument, we can conclude that every node $k$ must also be a source in exactly one of the $\widehat{G}(i)$, and thus the in-degree of each node is at most $1$. Therefore, each node must have in-degree $d_{\mathrm{in}}$ that satisfies $n-2 \leq d_{\mathrm{in}} \leq 1$, which is only possible if $n \leq 3$.

Among the 16 non-isomorphic $n=3$ graphs, there are only two graphs that have $ d_{\mathrm{in}} =n-2=1$ for all nodes in the graph, namely the $3$-cycle and the graph consisting of a $2$-clique with an outgoing edge to the third node. For the $3$-cycle, it is straightforward to check that every node-complement graph does indeed have a source and a target.  For the $2$-clique with the outgoing edge, however, complementing with respect to the third node yields a graph with no source.  Thus, if $G$ has the property that every node-complement graph contains a source and a target, then $G$ must be the $3$-cycle or have size $n \leq 2$.  
\end{proof}

We can now prove Theorem~\ref{thm:flexible-motifs}.

\begin{proof}[{\bf Proof of Theorem~\ref{thm:flexible-motifs}}]
Recall that part 2 of Theorem~\ref{thm:flexible-motifs} characterizing invariant forbidden motifs was already proven in Proposition~\ref{prop:inv-forbidden}.  To complete the proof of Theorem~\ref{thm:flexible-motifs}, we will next prove part 3 characterizing invariant permitted motifs, and then end with the proof of part 1 on flexible motifs.  

Putting together Lemma~\ref{lemma:n3} with Proposition~\ref{prop:node-complement}, we see that the only possible invariant permitted motifs are graphs of size $n \leq 2$ and the $3$-cycle. In \cite[Appendix A.4]{fp-paper}, $\FP(W)$ was determined for all graphs of size $n \leq 2$ (see Figure~\ref{fig:n2-graphs}).  Among these graphs, the invariant permitted motifs are precisely the singleton (panel A of Figure~\ref{fig:n2-graphs}), the independent set of size $2$ (panel B) and the $2$-clique (panel C). It now remains only to show that the $3$-cycle is also invariant permitted.

Let $W$ be any TLN matrix compatible with the $3$-cycle graph $G$. Then by Corollary~\ref{cor:singleton}, $\FP(W)$ does not contain any singleton fixed point supports, since $G$ does not contain any sinks.  Additionally, $\FP(W)$ cannot contain any fixed point supports of size $2$, since every size-$2$ subgraph is the directed edge, which is an invariant forbidden motif. By Lemma~\ref{lemma:parity} (parity), $|\FP(W)|$ is always odd, and thus $\FP(W)$ must have at least one fixed point. The only possibility is the full-support fixed point $123$. We thus conclude that the $3$-cycle is an invariant permitted motif with $\FP(W) = \{123\}$ for any $W$ compatible with $G$.

Finally, the characterization of flexible motifs in part 1 follows immediately from parts 2 and 3 since a motif is flexible if and only if it is neither invariant permitted nor invariant forbidden.  Moreover, Proposition~\ref{prop:fix-pt-construction} gives an explicit construction for a $W$ with $[n] \in \FP(W)$ for any graph that does not have both a source and a sink, while Proposition~\ref{prop:node-complement} gives a construction for a $W'$ with $[n] \not\in \FP(W')$ for any graph with $n \geq 3$ that is not the $3$-cycle.  Thus, for any graph that satisfies both these properties, we have explicit constructions demonstrating the flexibility of the motif.
\end{proof}

\section{Robust motifs and the proofs of Theorems~\ref{thm:robust-motifs} and~\ref{thm:collapse}}\label{sec:robust}

In this section, we first prove Theorem~\ref{thm:collapse}, which allows us to reduce $\FP(W)$ to that of a subnetwork $W_\tau$ whenever the corresponding graph $G_W$ contains both a source and a target. Recall from Theorem~\ref{thm:flexible-motifs} that these graphs are precisely the invariant forbidden motifs. We will then
use Theorem~\ref{thm:collapse}, together with results characterizing invariant permitted and forbidden motifs from the previous section, in order to prove Theorem~\ref{thm:robust-motifs}, which gives the full classification of robust motifs.

We begin with the following lemma, which is key to proving Theorem~\ref{thm:collapse}.  

\begin{lemma}\label{lemma:remove-sources}
 Let $G$ be a graph on $n$ nodes containing a source $s$ and a target $t$. Then
$$\FP(W) = \FP(W_{[n] \setminus s})$$
for every $W$ with graph $G_W=G$.
\end{lemma}

\begin{proof}
First we show that $\FP(W) \subseteq \FP(W_{[n] \setminus s})$.  Let $\sigma \subseteq [n]$ with $s \in \sigma$.  Since $t$ is a target of $G$, in particular it is a target of $G|_\sigma$, and so by Proposition~\ref{prop:target-source}, $\sigma \notin \FP(W)$ for any $W$ with graph $G_W=G$.  Thus, the only possible fixed point supports are $\sigma \subseteq [n]\setminus s$.  By part 2 of Lemma~\ref{lemma:inheritance}, a necessary condition for $\sigma \in \FP(W)$ is that it support a fixed point in every intermediate subnetwork containing it, and thus we must have $\sigma \in \FP(W_{[n] \setminus s})$.  Hence, $\FP(W) \subseteq \FP(W_{[n] \setminus s})$.

To show the reverse containment, we must show that every $\sigma \in \FP(W_{[n] \setminus s})$ survives the addition of the source $s$ to be a fixed point of the full network.  There are two cases: (1) $\sigma$ contains the target $t$, and (2) $\sigma$ does not contain $t$.  Observe that by Lemma~\ref{lemma:inheritance}, $\sigma \in \FP(W_{[n] \setminus s})$ is equivalent to $\sigma \in \FP(W_{\sigma \cup k})$ for all $k \in [n] \setminus s$; thus, to show that $\sigma \in \FP(W)$ all that remains is to show that we also have $\sigma \in \FP(W_{\sigma \cup s})$.  In case (1) where $t \in \sigma$ is an internal target, part (iii) of Proposition~\ref{prop:target-source} guarantees that $\sigma \in \FP(W_{\sigma \cup s})$ (since $G|_\sigma$ has not outgoing edges in $G|_{\sigma \cup s}$).  For case (2) where $t \notin \sigma$ is an external target, the fact that $\sigma \in \FP(W_{[n] \setminus s})$ guarantees in particular that $\sigma$ survived the addition of the target $t$, so $\sigma \in \FP(W_{\sigma \cup t})$, and thus by part (iv) of Proposition~\ref{prop:target-source}, $\sigma \in \FP(W_{\sigma \cup s})$ as well.  Thus we have $\FP(W) = \FP(W_{[n] \setminus s})$
for every $W$ with graph $G_W = G$.
\end{proof}

We can now prove Theorem~\ref{thm:collapse} by iteratively applying Lemma~\ref{lemma:remove-sources} as we ``strip off" sources in $G$.  

\begin{proof}[{\bf Proof of Theorem~\ref{thm:collapse}}]
Let $G$ be a graph that contains both a source $s$ and a target $t$. Let $\omega \:\dot\cup\: \tau = [n]$ be a source-target decomposition of $G$, i.e.\ $s \in \omega$, $t \in \tau$, $G|_\omega$ is a DAG, and there are no edges from $\tau$ to $\omega$. By a well-known property of DAGs, there exists a {\it topological ordering} of the nodes $1, \ldots, |\omega|$ such that if $j >i$, then $j \not\to i$ in $G|_\omega$ (so the edges in $G|_\omega$ only go from nodes that are lower in the ordering to nodes that are higher).  Additionally, since there are no edges from $\tau$ to $\omega$, node $1$ is not only a source in $G|_\omega$ but also a source in $G$.  Applying Lemma~\ref{lemma:remove-sources}, we see that
 $$\FP(W) = \FP(W_{(\omega \setminus \{1\}) \cup \tau})$$
 for any $W$ with graph $G_W = G$.  Now consider the graph $G|_{(\omega \setminus \{1\}) \cup \tau}$, where we have removed the source node $1$.  In this graph, node $2$ is a source, and so we can again apply Lemma~\ref{lemma:remove-sources} to obtain 
 $$\FP(W) = \FP(W_{(\omega \setminus \{1\}) \cup \tau}) = \FP(W_{(\omega \setminus \{1, 2\}) \cup \tau}). $$ 
Continuing in this fashion, we can remove all nodes from $\omega$, since they each become sources after all lower-numbered nodes have been removed. We thus obtain $\FP(W) = \FP(W_\tau)$.
\end{proof}

We can now prove Theorem~\ref{thm:robust-motifs}, which says that the invariant permitted motifs are robust, and every other robust motif belongs to one of the infinite families DAG1 or DAG2, depicted in Figure~\ref{fig:robust-motifs}.

\begin{proof}[{\bf Proof of Theorem~\ref{thm:robust-motifs}}]
Recall that every robust motif must be either invariant permitted or invariant forbidden. By part 3 of Theorem~\ref{thm:flexible-motifs}, all the invariant permitted motifs are robust and they are precisely the singleton, the $2$-clique, the independent set of size $2$, and the $3$-cycle. The proof of Theorem~\ref{thm:flexible-motifs} also established the $\FP(W)$ for each of these graphs.

It remains to show that all other robust motifs belong to either DAG1 or DAG2, and that all graphs in these infinite families are robust. Since all robust motifs that are not invariant permitted must be invariant forbidden, the remaining robust motifs must all contain both a source and target, by Theorem~\ref{thm:flexible-motifs}. Any such graph $G$ has a source-target decomposition $\omega \:\dot\cup\: \tau = [n]$, and thus by Theorem~\ref{thm:collapse}, $\FP(W) = \FP(W_\tau)$ for any $W$ with graph $G$. It follows from this partitioning that if $G$ is robust, then $G|_\tau$ must also be robust. Note that if $G|_\tau$ contains a source, then we can move it to $\omega$ and create a new source-target decomposition. Thus, by iteratively moving sources of $G|_\tau$ into $\omega$, we can obtain a partition where $G|_\tau$ has no sources (see Figure~\ref{fig:source-target-decomposition-example}). $G|_\tau$ must then be a robust motif that contains a target but has no sources. It follows that $G|_\tau$ must be an invariant permitted motif that contains a target $t$. There are only two such graphs: the singleton and the $2$-clique. If $G|_\tau$ is the singleton, so that $\tau = \{t\}$, then $G$ is in DAG1 and $\FP(W) = \FP(W_\tau) = \{t\}$. Otherwise, if $G|_\tau$ is the $2$-clique, so that $\tau = \{t,u\}$, then $\FP(W) = \{tu\}$. 
Thus, every robust motif that is invariant forbidden lives in DAG1 or DAG2, and this same reasoning shows that every graph in DAG1 and DAG2 is robust.  
\end{proof}

\bigskip

\noindent{\bf Acknowledgments.} This work was supported by NIH R01 EB022862 and NSF DMS-1516881. 

\bibliographystyle{unsrt}
\bibliography{CTLN-refs}

\begin{thebibliography}{10}

\bibitem{CTLN-preprint}
K.~Morrison, A.~Degeratu, V.~Itskov, and C.~Curto.
\newblock Diversity of emergent dynamics in competitive threshold-linear
  networks: a preliminary report.
\newblock Available at \verb!https://arxiv.org/abs/1605.04463!

\bibitem{book-chapter}
K.~Morrison and C.~Curto.
\newblock {\em Predicting neural network dynamics via graphical analysis},
  pages 241--277.
\newblock Book chapter in Algebraic and Combinatorial Computational Biology,
  edited by R. Robeva and M. Macaulay. Elsevier, 2018.

\bibitem{CTLN-paper}
K.~Morrison, C.~Langdon, C.~Parmelee, J.~Geneson, A.~Degeratu, V.~Itskov, and
  C.~Curto.
\newblock Dynamically relevant motifs in inhibition-dominated threshold-linear
  networks.
\newblock \emph{In preparation.}

\bibitem{fp-paper}
C.~Curto, J.~Geneson, and K.~Morrison.
\newblock Fixed points of competitive threshold-linear networks.
\newblock {\em Neural Comput.}, 31(1):94--155, 2019.

\bibitem{Seung-Nature}
R.~H. Hahnloser, R.~Sarpeshkar, M.A. Mahowald, R.J. Douglas, and H.S. Seung.
\newblock Digital selection and analogue amplification coexist in a
  cortex-inspired silicon circuit.
\newblock {\em Nature}, 405:947--951, 2000.

\bibitem{XieHahnSeung}
X.~Xie, R.~H. Hahnloser, and H.S. Seung.
\newblock Selectively grouping neurons in recurrent networks of lateral
  inhibition.
\newblock {\em Neural Comput.}, 14:2627--2646, 2002.

\bibitem{HahnSeungSlotine}
R.~H. Hahnloser, H.S. Seung, and J.J. Slotine.
\newblock Permitted and forbidden sets in symmetric threshold-linear networks.
\newblock {\em Neural Comput.}, 15(3):621--638, 2003.

\bibitem{flex-memory}
C.~Curto, A.~Degeratu, and V.~Itskov.
\newblock Flexible memory networks.
\newblock {\em Bull. Math. Biol.}, 74(3):590--614, 2012.

\bibitem{net-encoding}
C.~Curto, A.~Degeratu, and V.~Itskov.
\newblock Encoding binary neural codes in networks of threshold-linear neurons.
\newblock {\em Neural Comput.}, 25:2858--2903, 2013.

\bibitem{pattern-completion}
C.~Curto and K.~Morrison.
\newblock Pattern completion in symmetric threshold-linear networks.
\newblock {\em Neural Computation}, 28:2825--2852, 2016.

\bibitem{ReimannHessMarkram17}
M.W. Reimann, M.~Nolte, M.~Scolamiero, K.~Turner, R.~Perin, G.~Chindemi,
  P.~Dlotko, R.~Levi, K.~Hess, and H.~Markram.
\newblock Cliques of neurons bound into cavities provide a missing link between
  structure and function.
\newblock {\em Front Comput Neurosci}, 11:48, 2017.

\bibitem{ChambersMacLean16}
B.~Chambers and J.~N. MacLean.
\newblock Higher-order synaptic interactions coordinate dynamics in recurrent
  networks.
\newblock {\em PLoS Comput Biol}, 12(8):e1005078, 2016.

\bibitem{DecheryMacLean18}
J.~B. Dechery and J.~N. Mac{L}ean.
\newblock Functional triplet motifs underlie accurate predictions of
  single-trial responses in populations of tuned and untuned v1 neurons.
\newblock {\em PLoS Comput Biol}, 14(5), 2018.

\bibitem{BojanekZhuMacLean19}
K.~Bojanek, Y.~Zhu, and J.~Mac{L}ean.
\newblock Cyclic transitions between higher order motifs underlie sustained
  activity in asynchronous sparse recurrent networks.
\newblock Available at
  \verb!https://www.biorxiv.org/content/10.1101/777219v1.full!

\bibitem{num-directed-graphs}
The {O}n-line {E}ncyclopedia of {I}nteger {S}equences: number of directed
  graphs.
\newblock \url{http://oeis.org/A000273}, retrieved July 22, 2019.

\bibitem{num-DAGs}
The {O}n-line {E}ncyclopedia of {I}nteger {S}equences: number of acyclic
  digraphs.
\newblock \url{http://oeis.org/A003087}, retrieved July 22, 2019.

\bibitem{oriented-matroids-paper}
C.~Langdon, K.~Morrison, and C.~Curto.
\newblock Combinatorial geometry of threshold-linear networks.
\newblock In preparation.

\end{thebibliography}


\section{Appendix}

\subsection*{A.1 Strong domination implies general domination}
In this appendix, we briefly review the notion of \emph{general domination}, which was introduced in \cite{fp-paper} and shown to fully characterize the collection of fixed points of a TLN.  We then show that strong domination is a special case of this, and use this relationship to prove Lemma~\ref{lemma:strong-domination} from Section~\ref{sec:strong-domination}.

Motivated by Cramer's rule, we define
\begin{equation}\label{eq:s_i}
s_i^\sigma \od \det(I-W_{\sigma \cup i}; \theta) \quad \quad \text{for each } \sigma \subseteq [n] \text{ and } i \in [n].
\end{equation}
Note that Cramer's rule guarantees that when $\sigma \in \FP(W)$, the corresponding fixed point $x^*$ has $x_i^* = s_i^\sigma/ \det(I-W_\sigma)$ for all $i \in \sigma$.  

General domination captures the relationship between weighted sums of the magnitudes of the $s_i^\sigma$:

\begin{definition}[general domination]\label{def:general-domination}
Consider a TLN on $n$ nodes with matrix $W$, let $\sigma \subseteq [n]$ be nonempty, and let $\Wtil \od -I + W$.  For any $j,k \in [n]$, we say that $k$ {\em dominates} $j$ with respect to $\sigma$, and write $k >_{\sigma} j$, if
$$\sum_{i \in \sigma} \Wtil_{ki} |s_i^\sigma| > \sum_{i \in \sigma} \Wtil_{ji} |s_i^\sigma|.$$  
\end{definition}

Recall that strong domination has precisely the same set-up as general domination, but requires $\Wtil_{ki} \geq \Wtil_{ji}$ for all $i \in \sigma$ with at least one inequality being strict.  Thus, irrespective of the values of the $|s_i^\sigma|$, we see that whenever $k$ strongly dominates $j$ with respect to $\sigma$, we have general domination as well (although not vice versa).

The following theorem from \cite{fp-paper} shows that general domination precisely characterizes all fixed point supports of a TLN.  

\begin{theorem}[general domination --  Theorem 15 in \cite{fp-paper}] \label{thm:domination}
Consider a TLN on $n$ nodes with matrix $W$, and let $\sigma \subseteq [n]$. Then
$$\sigma \text{ is a permitted motif of $W$} \ \ \Leftrightarrow \ \ \sum_{\ell \in \sigma} \Wtil_{i\ell} |s_\ell^\sigma| = \sum_{\ell \in \sigma} \Wtil_{j\ell} |s_\ell^\sigma| \text{ for all } i, j \in \sigma.$$ 
If $\sigma$ is a permitted motif of $W$, then $\sigma \in \FP(W)$ if and only if
for each $k \not\in \sigma$ there exists $j \in \sigma$ such that $j >_\sigma k$.
\end{theorem}

The following lemma unpacks the consequences of Theorem~\ref{thm:domination} in the presence of different general domination relationships.  Note that parts (i)-(iii) of Lemma~\ref{lemma:domination} previously appeared as Lemma 16 in \cite{fp-paper}, but we provide the proofs again here, together with a proof the the new item (iv), for completeness.  

\begin{lemma}\label{lemma:domination}
Suppose $k >_\sigma j$. Then the following statements all hold:
\begin{itemize}
\item[(i)] If $j,k \in \sigma$, then $\sigma \notin \FP(W_\sigma)$, and thus $\sigma \notin \FP(W)$.
\item[(ii)] If $j \in \sigma$ and $k \notin \sigma$, then $\sigma \notin \FP(W_{\sigma \cup k})$, and thus $\sigma \notin \FP(W)$.
\item[(iii)] If $k \in \sigma$ and $j \notin \sigma$, then $\sigma \in \FP(W_{\sigma \cup j}) \; \Leftrightarrow \; \sigma \in \FP(W_\sigma)$.
\item[(iv)] If $j,k \notin \sigma$, then $\sigma \in \FP(W_{\sigma \cup k}) \Rightarrow  \FP(W_{\sigma \cup j}).$
\end{itemize}
\end{lemma}
\begin{proof}
For ease of notation, define $w_j^\sigma$ as the relevant weighted sum for domination:
$$w_j^\sigma \od \sum_{i \in \sigma} \Wtil_{ji} |s_i^\sigma|.$$
With this notation, Theorem~\ref{thm:domination} says that $\sigma$ is a permitted motif if and only if $w_i^\sigma = w_j^\sigma$ for all $i,j \in \sigma$. If $\sigma$ is not a permitted motif, then $\sigma \notin \FP(W_\sigma)$ and thus $\sigma \notin \FP(W)$.

The hypothesis $k>_\sigma j$ implies  $w_k^\sigma > w_j^\sigma$. Part (i) now follows immediately from Theorem~\ref{thm:domination}, which tells us that $\sigma$ is not a permitted motif if $j, k \in \sigma$. For part (ii), observe that if $\sigma$ is a permitted motif, then $w_i^\sigma = w_j^\sigma < w_k^\sigma$ for all $i \in \sigma$.  Thus there is no $i \in \sigma$ such that $i >_\sigma k$, and so $\sigma$ does not survive the addition of node $k$ by Theorem~\ref{thm:domination}. Hence, $\sigma \notin \FP(W_{\sigma \cup k})$.

Part (iii) is essentially the final statement of Theorem~\ref{thm:domination} (with the roles of $j,k$ reversed), and so $\sigma$ survives the addition of node $j$ precisely when $\sigma$ is a permitted motif.  
Finally, for part (iv), suppose $\sigma \in \FP(W_{\sigma \cup k})$.  Then by Theorem~\ref{thm:domination}, $\sigma$ is a permitted motif and there exists an $\ell \in \sigma$ such that $\ell >_\sigma k$, i.e.\ $w_\ell^\sigma > w_k^\sigma$.  Then we also have $w_\ell^\sigma > w_j^\sigma$, and so $\ell >_\sigma j$ and $\sigma \in \FP(W_{\sigma \cup j}).$
\end{proof}

Observe that Lemma~\ref{lemma:strong-domination} from Section~\ref{sec:strong-domination} is identical to the previous lemma except that its hypothesis requires that $k$ \emph{strongly} dominates $j$ with respect to $\sigma$.  Since strong domination implies general domination, we obtain Lemma~\ref{lemma:strong-domination} as an immediate corollary.


\subsection*{A.2 Survival of permitted motifs}\label{sec:survival}
The body of this paper has focused on characterizing graphs that are robust, and thus have $\FP(W)$ constant across all $W$ compatible with the graph.  A necessary condition for robustness is that for every subset $\sigma \subseteq [n]$, we have that $\sigma$ is either never in $\FP(W)$ or always in $\FP(W)$.  The first two conditions on $G|_\sigma$ given below characterize how $\sigma$ could be guaranteed to never be in $\FP(W)$, while the third condition characterizes precisely when $\sigma$ is guaranteed to always be in $\FP(W)$; thus, every subgraph $G|_\sigma$ of a robust motif must satisfy one of the following conditions: 
\begin{enumerate}
\item $G|_\sigma$ is an invariant forbidden motif,
\item $G|_\sigma$ is an arbitrary motif, whose embedding guarantees that whenever it is permitted, the fixed point always dies in $G$, or
\item $G|_\sigma$ is an invariant permitted motif, whose fixed point always survives in all of $G$.
\end{enumerate}
We have seen that the subgraphs satisfying condition 1 are precisely those that contain both a source and an internal target.  In Section~\ref{sec:strong-domination}, it was shown that if $G|_\sigma$ contains a source and is embedded in such a way that $G$ contains an external target of $G|_\sigma$, then any fixed point supported on $\sigma$ is guaranteed to never survive in all of $G$ (see Lemma~\ref{lemma:strong-domination}(ii) proven by strong domination).  Thus, any subgraph containing a source and embedded in $G$ with a target will satisfy condition 2, and thus be an allowable subgraph of a robust motif.  Corollaries~\ref{cor:cliques} and~\ref{cor:singleton} showed, again by strong domination, that the singleton and the $2$-clique embedded with no outgoing edges in $G$ are invariant permitted motifs that always survive to yield fixed points of $G$, thus satisfying condition 3.  A careful examination of the graphs in DAG1 and DAG2, together with the size-$2$ independent set and the $3$-cycle, shows that these are the only subgraphs that ever appear in robust motifs.  

In this appendix, we explicitly show that every other subgraph/possible embedding has \emph{parameter-dependent survival} and thus cannot be guaranteed to either always survive or never survive to support a fixed point in the full network. 

\begin{definition}
Let $G$ be a graph on $n$ nodes and $\sigma \subsetneq [n]$.  We say that $\sigma$ has \emph{parameter-dependent survival in $G$} if there exists a $W$ with graph $G_W=G$ such that $\sigma \in \FP(W)$ and another $W'$ with graph $G_{W'}=G$ such that $\sigma \in \FP(W'|_\sigma)$ but $\sigma \notin \FP(W')$.  In other words, $\sigma$ has parameter-dependent survival if there exists a choice of $W$ where $\sigma$ survives to support a fixed point of the full network and another choice where it is permitted, but it does not survive.
\end{definition}

Since the notion of a motif surviving in a network requires that it be permitted in its own restricted subnetwork, we do not consider invariant forbidden motifs since these are never permitted for any choice of $W$.  Thus, we restrict focus to motifs that do not contain both a source and a target; recall from Proposition~\ref{prop:fix-pt-construction} that every such motif was shown to have some parameter regime in which it is permitted.  The following proposition further shows that these motifs all have parameter-dependent survival irrespective of their embeddings in $G$ except for two special cases where strong domination applies (described above), forcing guaranteed survival or death.

\begin{proposition} \label{prop:param-dep-survival}
Let $G$ be a graph on $n$ nodes and $\sigma \subsetneq [n]$ such that $G|_\sigma$ does not contain both a source and an internal target.  Then
\begin{itemize}
\item[1.] If $G|_\sigma$ has no sources and no targets, then $\sigma$ has parameter-dependent survival in $G$.
\item[2.] \begin{itemize}
\item[(a)] If $G|_\sigma$ has a source, but no external target in $G$, then $\sigma$ has parameter-dependent survival in $G$.
\item[(b)] If $G|_\sigma$ has a source and an external target in $G$, then $\sigma \notin \FP(W)$ for any $W$ with graph $G_W=G$.  
\end{itemize}
\item[3.] \begin{itemize}
\item[(a)] If $G|_\sigma$ has an internal target and there is at least one outgoing edge from $\sigma$ in $G$, then $\sigma$ has parameter-dependent survival in $G$.
\item[(b)] If $G|_\sigma$ has an internal target and there are no outgoing edges from $\sigma$ in $G$, then $\sigma \in \FP(W) \ \Leftrightarrow \ \sigma \in \FP(W_\sigma)$ for every $W$ with graph $G_W = G$.
\end{itemize}
\end{itemize}
\end{proposition}

\begin{proof}
Part 2(b) follows directly from part (ii) of Proposition~\ref{prop:target-source}, while part 3(b) was covered by part (iii) of that proposition.  For all other parts of the statement, we will show parameter-dependent survival by exhibiting a choice of $W$ where $\sigma$ survives and a choice of $W'$ where $\sigma$ is permitted, but dies.  Recall from Section~\ref{sec:background} that when $\sigma$ is a permitted motif of $W$, it survives to yield a fixed point of the full network precisely when fixed point condition (2) is satisfied:
\begin{equation} \label{eq-off-condition}
\sum_{i\in\sigma} W_{ki}x_i^\sigma+\theta \leq 0 \text{ for all } k \notin \sigma.
\end{equation}

In the proof of Proposition~\ref{prop:fix-pt-construction}, it was shown that for any graph $G|_\sigma$ that does not have both a source and a target, it is possible to construct a $W_\sigma$ compatible with $G|_\sigma$ such that 
$$x_\sigma = \dfrac{\theta}{|\sigma|+r}{\bf 1_\sigma}$$
is a fixed point of $W_\sigma$, where $0 < |r|<|\sigma|$, and $r$ can be chosen so that $r >0$ if $G|_\sigma$ has no targets, or $r <0$ if $G|_\sigma$ has no sources.  For such a choice of $W_\sigma$ and the resultant $x_\sigma$, condition~\eqref{eq-off-condition} becomes
\begin{equation} \label{eq-off-condition2}
\sigma \in \FP(W) \ \ \Leftrightarrow \ \  \sum_{i\in\sigma} W_{ki} \leq - |\sigma| - r \ \text{ for all } k \notin \sigma. 
\end{equation}

For part 1, suppose $G|_\sigma$ has no sources and no targets.  First we will construct $W$ such that $\sigma \in \FP(W)$.  Choose $r<0$ and construct $W_\sigma$ so that $x_\sigma = \frac{\theta}{|\sigma|+r} {\bf 1}_\sigma$ is a fixed point of $W_\sigma$.  For each $k \notin \sigma$, choose 
$$W_{ki} \leq -1 - \frac{r}{|\sigma|}$$ 
in a way that is compatible with the edges from $\sigma$ to $k$ in $G$ (note compatibility with edges is always possible since $-1 - \frac{r}{|\sigma|} >-1$ for $r<0$).  Then $\sum_{i\in\sigma} W_{ki} \leq - |\sigma| - r$ for all $k \notin \sigma$, and so $\sigma \in \FP(W)$.  

Next to construct $W'$, choose $r>0$ and $W'_\sigma$, so that $x_\sigma = \frac{\theta}{|\sigma|+r} {\bf 1}_\sigma$ is a fixed point of $W'_\sigma$. For some $k \notin \sigma$, set 
$$W'_{ki} > -1 - \frac{r}{|\sigma|}$$ 
in a way that is compatible with the edges from $\sigma$ to $k$ in $G$ (here this is possible because $-1 - \frac{r}{|\sigma|} <-1$ for $r>0$).  Fill in the remaining entries of $W'$ arbitrarily in a way consistent with the graph $G$.  Since $\sum_{i\in\sigma} W'_{ki} > - |\sigma| - r$ for at least one $k$, we have $\sigma \notin \FP(W')$.\\

For part 2a, suppose $G|_\sigma$ has a source, but no external targets in $G$.  Because of the source, we must choose $r>0$ in the construction of $W_\sigma$. For each $k \notin \sigma$, there exists some $j \in \sigma$ such that $j \not\to k$, since $k$ is not an external target.    Set $W_{kj} \leq -1 - 2r$ while for all other $i \in \sigma\setminus \{j\}$, set $W_{ki} \leq-1 +\frac{r}{|\sigma| -1}$ in a way that is compatible with the graph $G$.  Then $\sum_{i\in\sigma} W_{ki} \leq - |\sigma| - r$ for all $k \notin \sigma$, and so $\sigma \in \FP(W)$.  For $W'$, again choose $r>0$ and follow the same method used to construct $W'$ from part 1.  Then we have $\sigma \notin \FP(W')$.\\

Finally, for part 3a, suppose $G|_\sigma$ has an internal target and at least one outgoing edge in $G$.  To construct $W$, follow the same method used to construct $W$ in part 1 to show $\sigma \in \FP(W)$.  Because of the target, we must choose $r<0$ in the construction of $W'_\sigma$; in particular, choose $r$ so that $-\frac12 < r < 0$.  Since $G|_\sigma$ has at least one outgoing edge, there exists a $k \notin \sigma$ and a $j \in \sigma$ such that $j \to k$ in $G$.  Set $W_{kj} > -1 - 2r$ and for all other $i \in \sigma$, $i \neq j$, set $W_{ki} > -1 +\frac{r}{|\sigma|-1}$ in a way that is compatible with $G$.  Fill in the remaining entries of $W'$ in any way consistent with the graph $G$.  Since $\sum_{i\in\sigma} W'_{ki} > - |\sigma| - r$ for at least one $k$, we have $\sigma \notin \FP(W')$.
\end{proof}

From part 3b of Proposition~\ref{prop:param-dep-survival}, we see that the only case when we can guarantee $\sigma \in \FP(W)$ for every choice of $W$ compatible with $G$ is when $G|_\sigma$ is an invariant permitted motif (so $\sigma \in \FP(W_\sigma)$ for all $W$) that contains an internal target and has no outgoing edges in $G$.  Since the singleton and the $2$-clique are the only invariant permitted motifs that contain a target (see Theorem~\ref{thm:flexible-motifs}), we immediately obtain the following corollary, which further elucidates why all the graphs in DAG1 contain a singleton with no outgoing edges, which is the support of the unique fixed point of the network, while all the graphs in DAG2 contain a $2$-clique with no outgoing edges, and again this $2$-clique is the only fixed point support of the network.  

\begin{corollary}
Let $G$ be a graph on $n$ nodes and $\sigma \subseteq [n]$.  If $\sigma \in \FP(W)$ for every $W$ with graph $G_W=G$, then $G|_\sigma$ is either a singleton or a $2$-clique and there are no outgoing edges from $G|_\sigma$ to the rest of $G$.  
\end{corollary}

\end{document}